\documentclass{article}

\usepackage{arxiv}

\usepackage[utf8]{inputenc} % allow utf-8 input
\usepackage[T1]{fontenc}    % use 8-bit T1 fonts
\usepackage{hyperref}       % hyperlinks
\usepackage{url}            % simple URL typesetting
\usepackage{booktabs}       % professional-quality tables
\usepackage{amsthm} 
\usepackage{amsfonts,amsmath,amssymb}       % blackboard math symbols
\usepackage{nicefrac}       % compact symbols for 1/2, etc.
\usepackage{microtype}      % microtypography
\usepackage{lipsum}		% Can be removed after putting your text content
\usepackage{graphicx}
\usepackage{natbib}
\usepackage{doi}

\usepackage[T1]{fontenc}
\usepackage{algorithm2e}
\usepackage{algorithmic}
\usepackage{enumitem}

\SetKwProg{func}{Function}{:}
\DontPrintSemicolon
\usepackage{graphicx}
\usepackage{graphicx}
\usepackage{todonotes}
\newtheorem{theorem}{Theorem}
\newtheorem{lemma}[theorem]{Lemma}
\newtheorem{example}{Example}
\title{Chrisimos: A useful Proof-of-Work for finding Minimal Dominating Set of a graph}
\author{Diptendu Chatterjee\\
Computer Science and Information Systems Department\\ BITS Pilani, K K Birla Goa Campus, India\\	
\href{mailto:diptenduc@goa.bits-pilani.ac.in}{diptenduc@goa.bits-pilani.ac.in}\\
\And
Prabal Banerjee\\
Avail and Indian Statistical Institute Kolkata, India\\ 
\href{mailto:mail.prabal@gmail.com}{mail.prabal@gmail.com}\\
\And
Subhra Mazumdar\\
TU Wien, and CDL-BOT, Vienna, Austria\\
\href{mailto:subhra.mazumdar@tuwien.ac.at}
{subhra.mazumdar@tuwien.ac.at}
}
\begin{document}
\maketitle 
\begin{abstract}
 Hash-based Proof-of-Work (PoW) used in the Bitcoin Blockchain leads to high energy consumption and resource wastage. In this paper, we aim to re-purpose the energy by replacing the hash function with real-life problems having commercial utility. We propose \emph{Chrisimos}, a useful Proof-of-Work where miners are required to find a \emph{minimal dominating set} for real-life graph instances. %Each miner is free to choose any method that outputs a \emph{minimal dominating set} of the given graph instance in polynomial time. We propose a method for estimating the time to generate each block, called block interval time. 
 A miner who is able to output the \emph{smallest dominating set} for the given graph within the block interval time wins the mining game. We also propose a new chain selection rule that ensures
the security of the scheme. Thus our protocol also realizes a decentralized \emph{minimal dominating set} solver for any graph instance. We provide formal proof of correctness and show via experimental results that the block
interval time is within feasible bounds of hash-based PoW.
\end{abstract}

\section{Introduction}
Cryptocurrencies have revamped the Banking and Financial Sector. Parties can execute transactions without relying on any trusted entity. Blockchain forms the backbone of these cryptocurrencies or tokens, and any transaction added to the blockchain is secure and tamper-proof \cite{xiao2020survey}. In permissionless blockchains, anyone can join and partake in the consensus without needing any approval \cite{tasca2017taxonomy}. 

Permissionless Blockchains are prone to Sybil attacks. An attacker creates an artificially large number of fake identities \cite{douceur2002sybil} with the intention of controlling the network. Thus, permissionless blockchains use a Sybil-resistant mechanism called Proof-of-Work (PoW) \cite{nakamoto2008bitcoin} that capitalizes on the computational resource invested by a participant to reach a consensus. PoW uses cryptographic primitives like hash hence the name hash-based PoW. Hash functions are computationally difficult with efficient verification. This makes PoW-based mining computationally expensive. Bitcoin is heavily criticized for the enormous energy consumption needed for mining \cite{9741872}. As per the report in January 2023, \cite{energy}, Bitcoin is estimated to have an annual power consumption of 127 \texttt{terawatt-hours (TWh)} - more than many countries, including Norway. Alternate consensus protocols overcome this problem by quantifying work done based on the utilization of other resources that have no adverse impact on nature. However, an alternate consensus has its own problems. Proof-of-Stake \cite{king2012ppcoin} and Proof-of-Authority \cite{poa} eliminate democracy and move the power to the ``richest in the room''. Proof-of-Space \cite{park2018spacemint} needs more storage space when more miners are added to the network. Proof-of-Spacetime \cite{moran2019simple} requires users to lock up a certain amount of coins in order to store data on the Blockchain, which can create an entry-level barrier for new users. Proof-of-Burn \cite{karantias2020proof} wastes coins as mining power is proportional to the amount of money a participant is willing to burn. 

\textbf{Previous works on useful PoW}. Instead of changing the underlying consensus mechanism, we analyze the literature of Proof-of-Useful-Work (PoUW) \cite{ball2017proofs}. PoUW utilizes the energy of the miners to perform some useful task that has either commercial utility or academic purpose. Such systems generally focus on NP-complete problems. This class of problems perfectly fits into blockchain systems because identifying the solution
in the first place has no known polynomial algorithm but the solution to the problem can be
verified in polynomial time \cite{oliver2017proposal}. Early designs and implementations like Primecoin \cite{king2013primecoin} and Gapcoin \cite{gap} are based on number-theoretic hardness but these protocols do not solve any problem of utility. Further constructions for PoUW mining were given by Loe et al. \cite{loe2018conquering}, and Dotan et al. \cite{dotan2020proofs}. In all these previous approaches, the security of the system was not rigorously analyzed and in many cases, the attacker can manipulate the system by providing easy instances to solve. A hybrid scheme for PoUW based on the artificial instance of the traveling salesman problem (TSP) has been proposed in \cite{syafruddin2019blockchain}, where the utility of artificial TSP is questionable. Another hybrid approach to mining \cite{philippopoulos2020difficulty}  combines hash value calculations with difficulty-based incentives for problem-solving. However, miners can be dishonest, find multiple solutions to a real-life instance, and instead of using the best one, they mine several blocks in a row using all the solutions. 
%Another PoUW protocol named proof-of-search \cite{shibata2019proof} is used as a tool to find approximate solutions to any optimization problem. However, the complexity of the overall proof-of-search protocol is too large, and the system is too restrictive in terms of providing a solution. 

Another PoUW protocol \emph{Ofelimos} \cite{fitzi2022ofelimos} is based on the doubly parallel local search (DPLS) algorithm. DPLS represents a general-purpose stochastic local-search algorithm having a component called the exploration algorithm. However, their approach does not take into account the difficulty of NP-hard problem instances which might lead to unfairness while block generation. A recent work called Combinatorial Optimization Consensus Protocol (COCP) \cite{todorovic2022proof} proposes efficient utilization of computing resources by providing valid solutions for the real-life instances of any combinatorial optimization problem. A miner upon finding the solution to a combinatorial optimization problem sends it to a solution pool controlled by a centralized entity. This is a major drawback because the entity controlling the solution pool could be malicious and tends to favor or sabotage a particular miner. 

\textbf{Choosing dominating set problem}. We address the drawbacks of the state-of-the-art by proposing a novel useful PoW. We do not claim to reduce energy consumption but replace the hash function with a graph-theoretic problem of finding a minimal dominating set. Dominating set of a graph has a lot of applications in real life \cite{sym12111885}. In Wireless Sensor Networks (WSNs), the operator selects a certain subset of nodes called backbone nodes. Backbone nodes form a minimal dominating set of WSNs that reduces the communication overhead, increases the bandwidth efficiency, and decreases the overall energy consumption \cite{asgarnezhad2011connected}.  Online social network sites like Facebook, LinkedIn, and X are among the most popular sites on the Internet. Due to the financial limitation of the budget, it becomes important to effectively select the nodes in such a network to realize the desired goal. Dominating sets have an important role in social networks to spread ideas and information within a group. It can be used for targeted advertisements, and alleviate social problems \cite{wang2014domination}.

\textbf{Contributions}. We summarize the contributions as follows:\\
(i) We propose \emph{Chrisimos}\footnote{In Greek, Chrisimos means useful}, a useful PoW where miners find a minimal dominating set of a real-life graph instance instead of generating nonce for hash. \\
(ii) We use the greedy heuristic for finding a minimal dominating set discussed in \cite{alon2016probabilistic} on publicly available graph datasets to estimate the time taken for mining a block using our proposed PoW. The estimated time corresponding to a graph instance is recorded in a publicly available lookup table. We use this table to estimate \emph{block interval time} for any new graph instance.\\ 
(iii) Any miner returning the smallest dominating set in a given block interval time wins the mining game. This induces competition among miners to return the best solution in the given block interval time. Thus, our protocol simulates a decentralized minimal dominating set solver.\\
(iv) We define a new chain selection rule where instead of choosing the longest chain \cite{nakamoto2008bitcoin}, we choose the chain that has the highest work done. We prove that our chain selection rule guarantees the security of the proposed scheme.\\
(v) We run the greedy heuristic on certain datasets and construct the lookup table for block interval time. We observe that the block interval time is within feasible bounds of hash-based PoW, demonstrating the practicality of our approach. %\pb{if there is space below, prob good to make the subheadings in the section like "Contributions" with a bit more line gap/other format}

 \subsection{Outline}
The rest of the paper is organized as follows: in Section \ref{background} we provide the background needed for comprehending our work, in Section \ref{solution} we provide a detailed description of \emph{Chrisimos}, correctness, and complexity analysis of the scheme are discussed in Section \ref{analysis}, we discuss a new chain selection rule in Section \ref{chain}, in Section 
\ref{security} we discuss the security of our proposed PoW scheme, experimental analysis in Section \ref{experiment} shows that \emph{Chrisimos} adds a block to the chain within a bounded time similar to hash-based PoW, and finally we conclude the paper in Section \ref{conclusion}.
\section{Background and Notations}
\label{background}
\subsection{Notations used}
We denote a graph instance as $G(V,E)$ where $V$ is the set of vertices and $E \subseteq V \times V$ are the edges connecting the vertices in the graph $G$, and $n=|V|$. For any vertex $v_i \in V$, $N_i=N(v_i)$ denotes set of neighbors of the vertex $v_i$ where $(v_i,w) \in E \iff w \in N(v_i) $. We denote $\delta=min_{v_i \in V}(N(v_i))$ as the minimum degree of $G$ and $\gamma=max_{v_i \in V}(N(v_i))$ as the maximum degree of $G$. $DS(G)$ denotes the dominating set of the graph $G$ and we define $k$ as the permissible limit of the dominating set as defined in \cite{alon2016probabilistic}. The terms related to Blockchain such as \emph{block} $\mathcal{B}$, \emph{transaction, transaction fee,} \emph{coinbase transaction}, \emph{Merkle root} $h_{MR}$, \emph{miner} $M$, \emph{block reward} and \emph{fork} are defined as follows:

(a) \emph{Block} $\mathcal{B}$: A grouping of transactions, marked with a timestamp, and a fingerprint of the previous block. The block header is hashed to find a proof-of-work, thereby validating the transactions. Valid blocks are added to the main blockchain by network consensus.

(b) \emph{Transaction}: A transaction is a message that informs the Bitcoin network that a transfer of ownership of bitcoins has taken place, allowing the recipient to spend them and preventing the sender from spending them again once the transaction becomes confirmed.

(d) \emph{Transaction Fee}: The amount of Bitcoins (BTC) included in transactions as encouragement for miners to add the transactions to blocks and have them recorded on the blockchain.

(b) \emph{Coinbase transaction}: The first transaction in each Bitcoin block, which distributes the subsidy earned when a miner successfully validates it as well as the cumulative fees for all transactions included in the block. The coinbase transaction effectively creates new bitcoin.

(c) \emph{Merkle root $h_{MR}$}: The Merkle root is the final node in a Merkle tree. It is a hash which includes all other hashes in the tree. If a single hash is altered within a tree, this change will ripple upwards, changing the Merkle root completely.
Merkle roots are used in Bitcoin as efficient commitments to large data sets. The Merkle root of all transaction IDs (txids) in a block is included in the block header, so that if the Merkle root changes, the Proof-of-Work is rendered invalid. This setup ensures that once a block is published, no transaction within the block can be altered

(d) \emph{Miner} $M$: A miner is someone who tries to create blocks to add to the blockchain (the term also refers to a piece of software that does this). Miners are rewarded for their work by the Bitcoin protocol, which automatically assigns new bitcoins to the miner who creates a valid block. This is how all bitcoins come into existence.

(e) \emph{Block reward}: The reward, including the mining subsidy and transaction fee, that miners receive for successfully mining a Bitcoin block.

(f) \emph{Fork}: A change to the Bitcoin protocol that can come as a codebase fork, blockchain fork, hard fork, or soft fork, with one version of the protocol ``forking"" off from another. Forks can also include codebase forks and blockchain forks, though these are not necessarily protocol changes.

A graph $G$ transformed using transaction set $\mathcal{T}$ is denoted as $G_T(V_T,E_T)$ where $|V_T|=2|V|$, $|V|$ denotes cardinality of set $V$ . 
%\subsection{Bitcoin Transactions}
% We state the necessary background required for comprehending our work. If we mention network then it refers to the Bitcoin P2P network and if we mention graph then it refers to the input instances on which the miner needs to find a minimal dominating set. We use the term node (or nodes) to denote a vertex (or vertices) of a graph.

\subsection{Hash-based PoW}
 Bitcoin P2P network uses Adam Back’s Hashcash \cite{back2002hashcash} for PoW. A miner constructs a block that has a transaction and a block header. A difficulty target is specified in the block header that defines the required PoW difficulty to make this a valid block. A miner has to find a nonce value, such that the result of applying the \texttt{SHA256} hash function to a tuple consisting of her public key, the hash of the previous block, the hash of the
current block to be added, and the nonce value, satisfies the difficulty target. The hash function is denoted as $\mathcal{H}$. Finding the nonce is akin to rolling a many-sided die and getting the result. 
Given the one-way property of hash functions, the only strategy for finding the valid nonce is to repeatedly try randomly-generated nonce until one of them solves the problem. At least one miner must have found the desired nonce whenever a block gets added to the chain \cite{antonopoulos2014mastering}. On average, it takes about ten minutes to construct a block in the Bitcoin network. The difficulty target is defined as the number of zeros that are required for the hash to be a valid solution and it is updated regularly to ensure that the time to construct a block remains nearly equal to ten minutes. Miners are rewarded with the block mining fee upon finding the correct nonce.

%Since mining takes a nontrivial amount of computation and the use of specialized hardware, it may not be possible for an individual miner to find the nonce on her own. Usually, miners join any existing mining pool. If any member of the pool succeeds in finding the nonce, others earn a fraction of the block reward based on the computation power invested. The mining pool is administrated by a miner called a pool operator. The pool operator is solely responsible for the distribution of jobs to individual miners as well as for sharing the reward if the pool succeeds in adding a block to the chain \cite{antonopoulos2014mastering}. %This ensures a steady flow of earnings even though the miner solved the problem partially. 

\subsection{Dominating set problem} 

A dominating set $DS(G)$ of a graph $G(V, E)$ can be defined as a subset of vertices $V': V' \subseteq V$, and every vertex in $G$ is either in $V'$ or is adjacent to some vertex in $V'$ \cite{10.1007/978-3-540-30140-0_19}.

A set $DS(G)$ is minimal dominating set of the graph $G$ if it does not contain any other dominating set as a proper subset. A dominating set of $G$ of the lowest cardinality is called \emph{minimum dominating set} (MDS). Computing an MDS is NP-hard. The decision version of the problem i.e. \emph{dominating set problem} is NP-complete. It is defined as \emph{``Given a graph $G(V,E)$ and an integer $k$, does $G$ have a dominating set of size less than $k$?''}\cite{garey1979computers}

Since dominating set problem is an NP-complete problem, we provide proof of the existence of a dominating set within a feasible bound for any graph. This proof ensures that any miner will be able to find the solution within the given bound.
\begin{theorem}
\label{th1}
For a graph $G$ with $n$ vertices and minimum degree $\delta$ there exists a dominating set of size less than $k=\frac{n(1+\ln(1+\delta))}{1+\delta}$ \cite{alon2016probabilistic}.
\end{theorem}
\begin{proof}
The proof is provided in \cite{alon2016probabilistic} and we use it here for justifying the existence of the solution. Given the extended graph $G_T(V_T,E_T)$ with $n'=|V_T|$ vertices and minimum degree $\delta'$, 
 we create a dominating set $X$ of $G_T$ by randomly selecting vertices from $V_T$ with probability $p$. The $E(|X|)=n'p$. Now if we look into the set of vertices in $Y=V_T\setminus X$, then the probability for a vertex $y\in Y$ not to be covered by $X$ is at most $(1-p)^{1+\delta'}$ as the vertex $y$ and all its neighbors (at least $\delta'$ in number) can not be a neighbor of any of the vertices in $X$. We include all such vertices in a set $Y_X$. So, $E(|Y_X|)=n'(1-p)^{1+\delta'}$ Naturally, $X \cap Y_X= \phi$ and $X \cup Y_X$ give a dominating set of $G_T$. Here $|X \cup Y_X|= n'p+n(1-p)^{1+\delta'}$. Now we set the value of $p$ to minimize the size of this dominating set. 
\begin{equation}
\label{oureq}
    |X \cup Y_X|=n'p+n'(1-p)^{1+\delta'} \approx n'p+n'e^{-p(1+\delta')}
\end{equation}
We denote the dominating set of $G_T$ as $DS(G_T)=|X \cup Y_X|$. Taking first-order differentiation of $DS(G_T)$ w.r.t $p$,
\begin{equation}
    \frac{\mathrm{d}DS(G_T)}{\mathrm{d}p}=n'(1-(1+\delta')e^{-p(1+\delta')})=0 \implies p=\frac{\ln(1+\delta')}{1+\delta'}
\end{equation}
Substituting $p=\frac{\ln(1+\delta')}{1+\delta'}$ in Eq.\ref{oureq}, the bound on the dominating set is $\frac{n'(1+\ln(1+\delta'))}{1+\delta'}$.
\end{proof}

Choosing a small bound $k$ on the size of dominating set allows a utility company to get a good solution for a given budget. The lower the size of dominating set, the more effective it is in terms of the cost of implementing a certain objective in the graph instance. 

%Influencers can be targeted to inculcate healthy habits in the graph or for spreading information.

%   \textit{Minimum Dominating Set:} Minimum Dominating set of a graph $G$, is the one which has the minimum cardinality among all the dominating sets of $G$.
% Now, it is known that, finding minimum dominating set of any given graph is NP-Hard.
% Although using probabilistic method, we can find a domination set of $G(V,E)$ of size $d$, where $d < n.\frac{1 + ln(\delta+1)}{\delta + 1}$, $\delta$ being the minimum degree of graph and $|V|=n$. This is the permissible size of dominating set in graph $G$.
% Instead of this specific problem, can we use a class of problems? Can we pin down the exact features why this problem is a good fit. 

\section{Our Proposed Solution}
\label{solution}
% In this section, we describe the protocol where the miner needs to find out the dominating set of a graph instance instead of finding a nonce for a hash.

\subsection{System and adversarial model}
There are three types of participants in our protocol: (i) any utility company supplying a graph instance, (ii) a committee selecting the graph instance, and (iii) miners of the Bitcoin network. Each miner can mine and verify the blocks. By any utility company, we mean it could be any social networking company or company providing telecommunication services. Such networks change quite frequently with time, so there remains a steady flow of input to the Bitcoin network. Utility companies earn high profits by utilizing the dominating sets to realize their objective. Thus, the company shares a portion of its profit as remuneration with the miner who solves the dominating set for the given graph instance. Any utility company submits its graph instances to a common public platform and the identity of the utility company is not revealed. The common public platform is managed by a committee, and we assume that the majority of the members are honest. The committee members checks whether the problem submitted is sufficiently hard by setting a threshold for hardness. They may use the framework proposed in \cite{sym15010140} to estimate the difficulty, and any problem whose difficulty lies below the threshold is rejected. A secret key $sk_C$ is shared among the committee members using distributed key generation protocol \cite{gennaro2007secure} and the corresponding public key is $pk_C$. These members select a graph instance randomly from the pool of problems, assign an identifier and sign the graph. At least \texttt{t-of-n} need to sign the graph instance where $t>\frac{n}{2}$ and the signature scheme used is universally unforgeable \cite{boneh2006strongly}. The identifiers are assigned in increasing order. So the latest graph will have an identifier higher than the previous instances. A representative of the public platform announces the graph instance to the Bitcoin network. Instead of storing the graph instance, miners can download the graph instances from that platform. 

We assume that any adversary has bounded computation power. We consider a synchronous communication model, i.e., if a message is sent at time $t$, it reaches the designated party by $t+\Delta$, where $\Delta$ is the upper bound within which a transaction is confirmed in the blockchain and becomes visible to others.

\subsection{Phases of the protocol}
We define the different phases of \emph{Chrisimos}:\\
(i) \emph{Transaction Selection for Block}: A miner creates a block $\mathcal{B}$ with a set of valid transactions, the coinbase transaction and reward transaction set by the utility company to create set $\mathcal{T}$.\\
(ii) \emph{Preprocessing Phase}: The committee members decide on a graph instance $G$, assigns an id $id_G$, generates hash of the graph $\mathcal{H}(G)$ and signs it using $pk_C$. The network receives a graph $G$, along with the tuple having graph id ${id}_G$, hash of the graph $\mathcal{H}(G)$, minimum degree of the graph $\delta$, block interval time $T^G_{max}$ within which a new block must be added to the chain, and signature on this hash $\mathcal{H}(G)$ as input. The signature on the hash ensures that the graph instance is supplied from the legitimate public platform.
\begin{itemize}
    \item[a.] \emph{Verification of Input}: A miner $M$ verifies the source of the graph and checks the correctness of $G$ by comparing it with the signed hash of the graph. Additionally, it checks if the id ${id}_G$ is higher than the graph id of the previous instance solved and added in the last block $\textrm{prev}_\mathcal{B}$. 
    \item[b.] \emph{Input Transformation}: Given the graph $G$, $M$ transforms it to $G_T(V_T,E_T)$. The transformation is dependent on $\mathcal{T}$ and the hash of $\textrm{prev}_\mathcal{B}$. 
   
\end{itemize}
(iii) \emph{Mining of Block}: $M$ finds the dominating set of this extended graph $G_T$. It adds the dominating set to the block $\mathcal{B}$ and broadcasts it to the rest of the network. \\(iv) \emph{Block Verification}: A verifier checks if the cardinality of the dominating set in $\mathcal{B}$ is less than the solution already stored in the given block interval time. If so, the verifier updates the last stored block to $\mathcal{B}$ provided the dominating set of $G_T$ is valid. Once $T^G_{max}$ expires, the miner adds the last stored block to the blockchain. 
%The miner will ignore this solution if the cardinality of the dominating set is higher than the 

%\s 
%A utility company shares a graph instance with the Blockchain network. The objective of the miners is to find the best possible dominating set for the graph.
We discuss the phases (ii.b.) \emph{Input Transformation}, (iii) \emph{Mining of the Block} and (iv) \emph{Block Verification} in detail. %We also discuss the reward provided to the miner and how a fork in the blockchain is resolved.
%and (iii) Block Reward after $k$ confirmations
% \subsubsection*{(I) Block Mining}
% \sm{Add some connecting statements}

\subsubsection*{Input Transformation}
We propose a rule for the graph transformation such that it is dependent on the transaction set of the block. Since all miners work on the same instance in a block interval time, if a miner has fetched the solution for $G$, any other miner can steal this solution and add it to its block during block propagation. Our proposed transformation ensures no miner can manipulate the given graph instance or earn a reward like a free rider. No two transaction sets of a block are the same due to differences in the coinbase transaction. If a miner $M$ extends graph $G$ to $G_T$ using $\mathcal{T}$, another miner $M'$ will have a different transaction set $\mathcal{T}'$. Transforming $G$ would lead to $G_T' \neq G_T$ and $M'$ has to find a valid dominating set on $G_T'$.  

The transformation involves extending the graph instance $G$ by adding another $|V|$ number of vertices. The new vertices are connected to the graph $G$ based on a rule that depends on the Merkle root formed using $\mathcal{T}$ and the hash of the previous block $\textrm{prev}_{\mathcal{B}}$. The \emph{Extend} function is defined as follows and mentioned in Algorithm \ref{extend}:

\RestyleAlgo{ruled}
\begin{algorithm}[!htbp]
			\caption{{\sf Extend}}
			\label{extend}
			%\fontsmall
   \textbf{Input}: $G,h_1,h_2$\\ 
			$L\leftarrow K(h_1,h_2)$\\
                Sort $V$ in descending order of degree\\
               Introduce $W=\{w_1,w_2,\ldots,w_{|V|}\}$\\
               $j=0$\\
               $E_T=E$\\
                \For{$v \in V$}
                {
                   
                  \For{$i$ in $L[0:|N(v)|-1]$}
                  {
                      \If{L(i) is 1}
                      {
                          Add edge $e$ between $v$ and $w_{j+1}$\\
                          $E_T=E_T\cup \{e\}$
                      }
                  
                  }
                     $j=j+|N(v)|-1$\\
                     $L=L+|N(v)|$\\
                }
                $AdjList_{w}\leftarrow GetList(|V|,\delta)$\\
                $temp_j=j=2$\\
                $k=|W|-1$\\
                \For{$w \in W$}
                {
                    \For{$bit \in AdjList_w[0:k]$}
                    {
                         \If{bit is 1}
                         {
                             Connect $w$ and $w_{j}$ with edge $e$\\
                             Add $e$ to $E_T$\\
                         }
                         $j=j+1$
                    }
                    $AdjList_w=AdjList_w+k+1$\\
                    $k=k-1$\\
                    $temp_j=temp_j+1$\\
                    $j=temp_j$
                }
                $V_T=V \cup W$
                
			\Return{$G_T(V_T,E_T)$}
	\end{algorithm}

\begin{enumerate}[label=(\alph*)]
    \item Miner gets hashes $h_{\textrm{prev}\_\mathcal{B}}$ and $h_{MR}$, where $h_{\textrm{prev}\_\mathcal{B}}$ is the hash of the previous block and $h_{MR}$ is Merkle root formed using the transaction set $\mathcal{T}$. We propose a function $K$ that uses  $h_{\textrm{prev}\_\mathcal{B}}$ and $h_{MR}$ and extends the graph $G$ to $G_T$. The function $K$ is defined as follows: $K: \{0,1\}^{\lambda} \times \{0,1\}^{\lambda} \rightarrow \{0,1\}^{|2E|}$. The output has an equal number of 0's and 1's, and we state the steps for generating the output:
    \begin{itemize}
        \item[(i)] If $2|E|\leq \lambda$, then first $\lambda$ bits of $\mathcal{H}(h_{\textrm{prev}\_\mathcal{B}},h_{MR})$ is the output of $K(h_{\textrm{prev}\_\mathcal{B}},h_{MR})$.
        \item[(ii)] If $2\lambda \geq 2|E|>\lambda$, the output of $K(h_{\textrm{prev}\_\mathcal{B}},h_{MR})$ is defined as follows: first $\lambda$ bits is $\mathcal{H}(h_{\textrm{prev}\_\mathcal{B}},h_{MR})$, denoted as $L_1$, and rest of the $min(2|E|-\lambda,\lambda)$ bits are the $1's$ complement of $L_1$. 
        \item[(iii)] If $2|E|> 2\lambda$, we generate the output of $K(h_{\textrm{prev}\_\mathcal{B}},h_{MR})$ follows: the first $2\lambda$ bits are generated by following steps (i) and (ii), the rest $2|E|-2\lambda$ bits by concatentating the sequence of $2\lambda$ bits for $\lfloor \frac{2|E|}{2\lambda} \rfloor$ times. The last block has the first $|2E|-(2\lambda\lfloor \frac{2|E|}{2\lambda} \rfloor)$ bits from the sequence of $2\lambda$ bits concatenated, so we have an output $L$ of size $2|E|$. Since $L$ has an approximately equal number of 0's and 1's, the number of edges connecting $V$ to $V_T\setminus V$ is $|E|$.
    \end{itemize}

    \iffalse
    \textcolor{red}{There is a pattern after an interval of $2k$. Do we need to argue whether this reduces the hardness of the problem? We will check this later.}
    \fi
    %\item Number the vertices of $G(V,E)$ where, $|V|=n$.
    \item The graph $G$ is represented as a set of $n$ ordered pairs, i.e., $G=\{(v_i,N_i)|v_i \in V \mbox{ and }  N_i= N(v_i)\}, \forall i \in \{1,\dots,n\}, \mbox{where} |N_i| \geq |N_j|, i<j$ where $N(v_i)$ is the set of neighbors of $v_i$. This implies that the vertices of $G$ are sorted in decreasing order of degree, where $v_1$ is the vertex with the maximum degree, followed by $v_2$ whose degree is the second maximum, and this ends with the vertex $v_n$ having the minimum degree.
    \item We introduce $|V|$ new vertices labeled $w_i,\forall i \in \{1,\dots,|V|\}$. These vertices form part of the extended graph $G_T(V_T,E_T)$ where $V_T=2|V|$, but they are assigned to the set $V_T \setminus V$. 
   % \item Include all the edges in $E$ to $E'$.
    \item The miner connects vertices in $V$ and $V_T\setminus V$, where  connections of $w_i$ with neighbours of $v_i, 1 \leq i \leq |V|$ is dictated by specific $|N_i|$ bits of $L$ ranging from $[\sum_{m=1}^{i-1}|N_m|+1, \sum_{m=1}^i |N_m|]$, the connection rule is as follows:\emph{ For $1\leq i \leq |V|$, if the $l^{th}$ bit among these $|N_i|$ bits of $L$ is $1$, then $w_i$ will have an edge with the $l^{th}$ element of the array $N_i$, where $N_i$ has the neighbors of $v_i$ as its element. If the $l^{th}$ bit is $0$, then $w_i$ is not connected to $l^{th}$ neighbor of $v_i$.}

   \item \emph{Procedure GetList($|V|,\hat{\delta}$)}: To connect the vertices $w$ in $V_T\setminus V$, probability that an edge exists between pair of vertices is $p>\frac{\delta}{n}$ where $\delta$ is the minimum degree in $G$. We justify the constraint on $p$ in Lemma \ref{c1}. We choose $p=\frac{\hat{\delta}}{n}$ where $\hat{\delta}=2\delta$ and $\hat{\delta}\in \mathbb{N}$. The miner must therefore generate an adjacency list, where at least $p {n \choose 2}$ edges exist. It is possible for the miner to extend the graph from $G$ to $G_T$ and send it as part of the block header. However, a graph may have a size of a few $MBs$ and in Bitcoin, the block size itself has a size of 1 MB. Blowing up the size of the block header from 80 B to a few MBs will lead to larger-sized blocks and introduce propagation delay. Hence the miner must send minimum information needed to generate graph $G_T$. We define a rule that allows generating a bit-sequence of length ${n \choose 2}$ encoding adjacency list of graph $G_T\setminus G$. So first vertex $w_1$ will have a list encoding with the rest of $(n-1)$ vertices, i.e., $w_2,w_3,\ldots,w_n$, for $w_2$ the list size is $(n-2)$ encoding connections with $w_3,\ldots,w_n$ and so on. A bit 1 denotes the existence of an edge. The rule ensures that the number of 1's in the bits-sequence must be at least $p {n \choose 2}$. We discuss the steps by which $M$ can generate the connections between vertices in $G_T\setminus G$ as follows. 
   \begin{itemize}
   
      % \item Miner computes the merkle root using the transaction set $\mathcal{T}$. Let it be $\bar{h}$.  
      % \item Given $\bar{h}$ is a bit string, calculate the number of $1's$ in $\bar{h}$. Let it be $\gamma$.
       \item The entire bit sequence for connections is divided into chunks, each of the size $\lfloor \frac{n}{\hat{\delta}} \rfloor$ bits. For selecting the first chunk, the rule is :
   \begin{itemize}
     \item If $\hat{\delta}$ is odd, set first $\lfloor \frac{n}{\hat{\delta}} \rfloor-1$ bits as 0s and last bit 1, or $x=(\lfloor \frac{n}{\hat{\delta}} \rfloor-1 )'s 0|1$.
       \item If $\hat{\delta}$ is even, set first bit as 1 and rest $\lfloor \frac{n}{\hat{\delta}} \rfloor-1$ bits as 0s , or $x=1|(\lfloor \frac{n}{\hat{\delta}} \rfloor-1 )'s 0$. 
   \end{itemize}
    The second chunk of $\lfloor \frac{n}{\hat{\delta}} \rfloor$ bits are the mirror image of $x$, the third chunk of $\lfloor \frac{n}{\hat{\delta}} \rfloor$ bits is left cyclic shift of $x$, the fourth chunk is mirror image of third chunk, this goes on till we generate $n \choose 2$ bits. 
    \item The generalized rule is any $(2i+1)^{th}$ chunk is left cyclic shift of  $(2i-1)^{th}$ chunk, and $(2i+2)^{th}$ chunk is mirror image of $(2i+1)^{th}$ chunk till we get $\binom{n}{2}$ bits. Thus we have at least $\frac{\hat{\delta}(n-1)}{2}$ number of 1's in total. 
    \item We generalize the indices having bit $1$ in the bit-sequence depending on whether $\delta$ is even or odd:\\
    (a) Indices that have bit $1$ if $(\hat{\delta}\mod 2=1)$ are $(2m+1)*\frac{n}{\hat{\delta}}$ and $(2m+1)*\frac{n}{\hat{\delta}}+1$ where $0\leq m \leq \lfloor \frac{(n-1)\hat{\delta}-1}{2}\rfloor$.\\
    (b) Indices that have $1$ if $(\hat{\delta}\mod 2=0)$ are $1+2m*\frac{n}{\hat{\delta}}$ and $2(m+1)*\frac{n}{\hat{\delta}}$ where $0\leq m \leq \lfloor \frac{(n-1)\hat{\delta}-1}{2}\rfloor$.
     
   \end{itemize}

\end{enumerate}
The extend function returns the graph $G_T(V_T,E_T)$ where $n'=|V_T|=2|V|$ and $|E_T|=2|E|+p\binom{n}{2}=2|E|+\frac{\hat{\delta}(n-1)}{2}$. We also illustrate the steps of input transformation via an example.

\begin{figure*}[!hbt]
    \centering
   \includegraphics[scale=0.35]{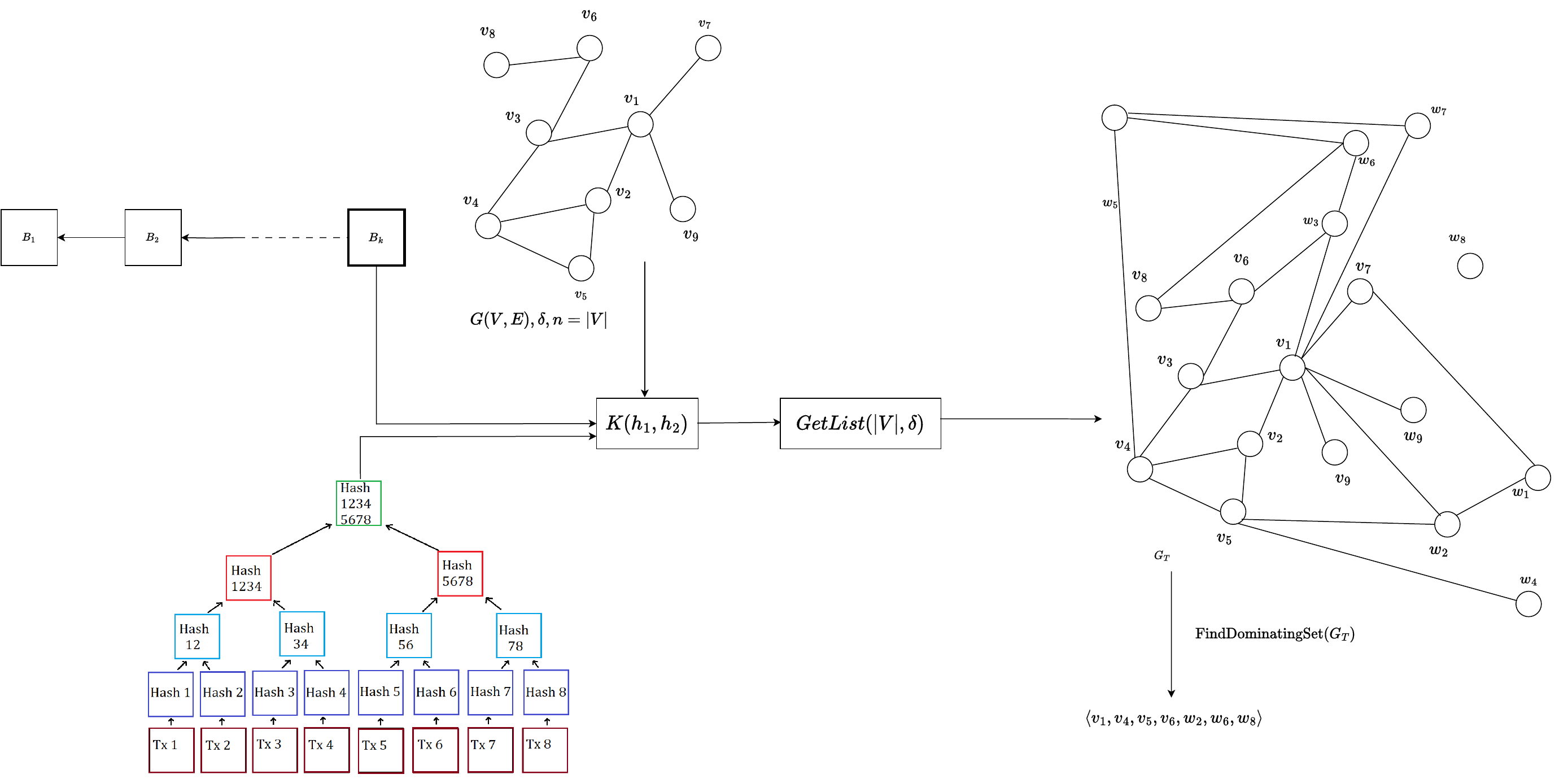}
    \caption{An instance of \emph{Chrisimos} for adding a new block}
    \label{fig:ds}
\end{figure*}
\begin{example}
    We illustrate our proposed mining procedure through Fig.\ref{fig:ds}. The block $B_{k+1}$ uses hash of block $B_k$ and Merkle root $12345678$ for transforming the input instance $G$ with 9 vertices $v_1,v_2,\ldots,v_9$ to $G_T$ with 18 vertices $v_1,v_2,\ldots,v_9,w_1,w_2,\ldots,w_9$. The miner outputs a minimal dominating set of $G_T$ as $\langle v_1,v_4,v_5,v_6,w_2,w_6,w_8 \rangle$. We discuss in detail the steps for transforming the input $G$ to $G_T$, estimating the block interval time $T_{max}^G$, and finally, retrieving a dominating set of $G_T$.
\end{example}

%We also illustrate the steps of input transformation via an example.

% \begin{figure*}[!hbt]
%     \centering
%     \includegraphics[scale=0.37]{dominating_diag.pdf}
%     \caption{An instance of \emph{Chrisimos} for adding a new block}
%     \label{fig:ds}
% \end{figure*}
% \begin{example}
%     We illustrate our proposed mining procedure through Fig.\ref{fig:ds}. The block $B_{k+1}$ uses hash of block $B_k$ and Merkle root $12345678$ for transforming the input instance $G$ with 9 vertices $v_1,v_2,\ldots,v_9$ to $G_T$ with 18 vertices $v_1,v_2,\ldots,v_9,w_1,w_2,\ldots,w_9$. The miner outputs a minimal dominating set of $G_T$ as $\langle v_1,v_4,v_5,v_6,w_2,w_6,w_8 \rangle$. We discuss in detail the steps for transforming the input $G$ to $G_T$, estimating the block interval time $T_{max}^G$, and finally, retrieving a dominating set of $G_T$.
% \end{example}

\subsubsection*{Mining of the Block}
The miner finds the dominating set of size at most the permissible size $k$ within a given block interval time $T_{max}^G$, pre-specified in the lookup table. The function ComputeBound on $G_T$ returns $k$, as defined in Theorem \ref{th1}. The pseudocode is defined in Algorithm \ref{gen_block}.

\RestyleAlgo{ruled}
\begin{algorithm}[!ht]
\SetKwRepeat{Do}{do}{while}
			\caption{{\sf Block Generation}}
			\label{gen_block}
			%\fontsmall
   \textbf{Input}: $G(V,E),{id}_G,H,\sigma_{H},pk_C,\mathcal{T},h_{\textrm{prev}\_\mathcal{B}}$\\
   %start\_time=current\_timestemp\\
   
			\If{ $H \neq \mathcal{H}(G)$ 
                or $SigVerify(\sigma_{H},H,pk_C)\neq 1$ or ${id}_G\leq GetPrevBlockGraphid()$}
                {
                    abort
                }
                Compute $h_{MR}=MerkleRoot(\mathcal{T})$\\
                Compute $G_T \leftarrow Extend(G,h_{\textrm{prev}\_\mathcal{B}},h_{MR})$\\
                $k=ComputeBound(G_T)$\\
                $DS'=V$\\
                
                  %end\_time=current\_timestemp\\
                
                \While{$ |DS'|> k$ and time elapsed $<T_{max}^G$}{
                    $DS(G_T)\leftarrow \textrm{FindDominatingSet}(G_T)$\\
                    \If{$|DS(G_T)| \leq k$}
                    {
                     
                        $DS'\leftarrow DS(G_T)$\\
                    }
                    
                     % end\_time=current\_timestemp\\
                
                }
                \If{time elapsed $<T_{max}^G$ and $ DS' \leq k$}
                {
                    $\textrm{block\_header}=HeaderGen(h_{\textrm{prev}\_\mathcal{B}},h_{MR},DS')$\\
                    $\mathcal{B}\leftarrow BlockGen(\textrm{block\_header},\mathcal{T})$\\
			         \Return{$\mathcal{B}$}
            }
            \Else
            {
                abort
            }
	\end{algorithm}

It broadcasts a block $\mathcal{B}$ containing the transaction set $\mathcal{T}$, the hash of the previous block $h_{\textrm{prev}\_\mathcal{B}}$, Merkle root of the transaction set $\mathcal{T}$ denoted as $h_{MR}$, address for fetching the graph $G$ along with the $id_G$, $\mathcal{H}(G)$, signature on $\mathcal{H}(G)$, degree $\hat{\delta}$, and the dominating set of $G_T$. 
\subsubsection*{Block Verification}
The verifier initializes the variable $past\_size_{DS}=2|V|$. When the miner receives the first block, it checks if the block has a valid dominating set. Cardinality of this dominating set is assigned to $past\_size_{DS}$. Now the verifier will cache this block until it gets a block with dominating set of sizes lower than $past\_size_{DS}$. The verifier will continue to check for new blocks till the $T_{max}^G$ expires. After this, it will accept the last block it had stored. Any solution appearing after $T_{max}^G$ is rejected. Miners who verify the block reach a consensus on the best result in the given epoch and a new block is added to the blockchain. 

When the verifier receives $\mathcal{B}$, it scans for the transaction set $\mathcal{T}$. From the block header, the verifier gets the dominating set $DS({G_T})$ of the extended graph $G_T$, id of the graph $G$ denoted as ${id}_G$, hash $\mathcal{H}(G)$, signature $\sigma_H$ on $\mathcal{H}(G)$, hash of the previous block $h_{\textrm{prev}\_\mathcal{B}}$, and Merkle root of the transaction set $\mathcal{T}$ denoted as $h_{MR}$. It downloads $G$ from the public domain using ${id}_G$, checks if graph $G$ provided is correct and whether ${id}_G$ is greater than the instance id of the graph solved previously. Verifier checks if $\mathcal{T}$ is correct, constructs the Merkle tree using $\mathcal{T}$, and checks the correctness of $h_{MR}$. The solution $DS({G_T})$ can be written as $V_{ds} \cup W_{ds}$. This is because some vertices of graph $G$, denoted as ${V}_{ds}$, and some vertices in $G_T\setminus G$, denoted as $W_{ds}$ will cover the graph.

The verifier checks whether the vertices in $V_{ds}$ cover the rest of the vertices $G$. If not all vertices are covered then the verifier sets the variable \emph{uncovered} to true. Next, the verifier calculates the indices that specify which vertices in $G_T\setminus G$ are connected to vertices in $G$ and checks how many vertices in $V_T\setminus V$ are also covered by $V_{ds}$. To calculate the indices, the verifier generates an output of size $\lambda$ by hashing over inputs $h_{\textrm{prev}\_\mathcal{B}}$ and $h_{MR}$. Let the output be $S$. The verifier can now infer the positions having bit 1 in the bit sequence of size 2$|E|$ if he knows the positions having bit $1$ in the bit sequence of size $\lambda$. We use the method $calc\_index$ on input $S$ for calculating the indices and store it in $Indices_S$. Once the verifier gets the indices, it will know which vertices in $V_T\setminus V$ got covered. We explain the function $calc\_index$ with an example. The verifier has to generate a binary sequence of size $\lambda$ and then use the rule explained in the example to get a bit string of size $2|E|$ with an equal number of 0's and 1's.

\RestyleAlgo{ruled}
\begin{algorithm}[!ht]
			\caption{{\sf Block Verify}}
			\label{verify}
			%\fontsmall
   \textbf{Input}: $\mathcal{B},past\_size_{DS}$ \\ 
			Parse $\mathcal{B}$ to get $h_{\textrm{prev}\_\mathcal{B}},h_{MR},{id}_G,H, \sigma_{H}$, $\delta$, dominating set $DS({G_T})={V_{ds}}\cup {W_{ds}}$ and $G\leftarrow get({id}_G)$ from block header and transaction set $\mathcal{T}$\\
   Set $visited\_set=\phi$\\
   
  \If{ $h_{MR} \neq MerkleRoot(\mathcal{T})$ or $H \neq \mathcal{H}(G)$ 
                or $SigVerify(\sigma_{H},H,pk_C)\neq 1$ or current time $\geq T_{max}^G$ or ${id}_G\geq GetPrevBlockGraphid()$ or ($past\_size_{DS}<|DS(G_T)|)$}
                {
                    reject solution                    
                }
                \Else
                {
                $S\leftarrow \mathcal{H}(h_{\textrm{prev}\_\mathcal{B}},h_{MR})$\\
               $Indices_S \leftarrow calc\_index(M)$\\
                \For{$v$ in ${V}_{ds}$}
                {
               
                    Mark $v$ as visited \\
                  
                         Add $v$ to $visited\_set$\\

                    \For{$v' \in N(v)$}
                    {
                         If $v'$ is not visited then mark it visited \\  
                         Add $v'$ to $visited\_set$\\
                    }
                    \For{ $v'' \in Indices_S(N(v))$}
                    {
                        Mark neighbor $v'' \in V_T\setminus V$ as visited if not visited and add $v''$ to $visited\_set$\\

                    } 
                
                }
                $AdjList_{w}\leftarrow GetList(|V|,\hat{\delta})$\\
                \For{$v$ in $W_{ds}$}
                {

                      Mark $v$ as visited and add $v$ to $visited\_set$\\

                      \For{ $w \in AdjList_w(N(v))$} 
                      {
                         %Add $w$ to $N(v)$ if the bit is 1  \\
                        Mark neighbor $w$ as visited if not visited and add $w$ to $visited\_set$\\

                      }

                      \For{$v$ in $V\setminus visited\_set$}
                      {
                      \For{ $v'' \in Indices_S(N(v))$}
                    {
                        Mark $v''$ as visited if not visited and add $v''$ to $visited\_set$\\
                              %Stop scanning $M$\\      
                    } 
                    %   Start scanning $M$ from the index marking the neighbor set of $v$
                    % {
                    %     \If{bit is 1 for a vertex $v: v \in DS\_{G_T}$}
                    %     {
                                         
                    %     }
                    % } 
                      }
                    }
                    \If{$|visited\_set|=2|V|$}
                    {
                            
                            $past\_size_{DS}=|DS({G_T})|$\\

                    }
                    \Else
                    {
                        reject solution
                    }
                    
                    }

	\end{algorithm}

\begin{example}
For ease of analysis, we consider $\lambda=|H(h_{\textrm{prev}\_\mathcal{B}}, h_{MR})|=10$. If we have \emph{2|E|} of size \emph{200}, and $\mathcal{H}(h_{\textrm{prev}\_\mathcal{B}},h_{MR})$ generates a \emph{10-bit} hash \emph{0101010110}, so at indices \emph{2, 4, 6, 8, 9} the bit is \emph{1}. The next block of \emph{10-bit} is the complement of the first block, and that will be \emph{1010101001}. As per the rule, if the position $i$ in this binary string $\mathcal{H}(h_{\textrm{prev}\_\mathcal{B}},h_{MR})$ is \emph{1}, then all the bits at position $m\lambda+i$ will be \emph{1}, where $m \mod 2=0$ and $m\lambda+i<2|E|, m \in \mathbb{N}, \lambda=10, 1\leq i \leq  \lambda$. So bits at positions \emph{2, 22, \ldots, 182}, \emph{4, 24,\ldots, 184}, \emph{6, 26,\ldots, 186}, \emph{8, 28,\ldots, 188}, \emph{9, 29,\ldots, 189} will be 1. So we get a total of \emph{50} indices where the bit is \emph{1}. We had mentioned that the bits from $\lambda+1$ to $2\lambda$ are the complement of the first $\lambda$ bits. Bit bit positions \emph{1, 3, 5, 7, 10} are \emph{1} in the output $\mathcal{H}(h_{\textrm{prev}\_\mathcal{B}},h_{MR})$. If a bit $i$ is 0 in this string, then the bit at position $z\lambda+i$ will be \emph{1} where $z \mod 2=1$. So bits at positions \emph{11, 31,\ldots, 191}, \emph{13, 33,\ldots, 193}, \emph{15, 35,\ldots, 195}, \emph{17, 37,\ldots, 197}, \emph{20, 40,\ldots, 200} will be \emph{1}. In total, we get \emph{100} indices where the bit is \emph{1}. 
\end{example}

If the variable \emph{uncovered} is true, the verifier checks if the vertices in $W_{ds}$ cover the rest of the uncovered vertices in $G$. After this step, the verifier checks if the remaining uncovered vertices in $V_T\setminus V$ get covered by checking the connection between the vertices in $V_T\setminus V$. A bit-sequence of length ${n \choose 2}$ encoding adjacency list of graph $G_T\setminus G$ has been discussed previously. Since the number of 1's in the bit string is $\hat{\delta} \frac{(n-1)}{2}$, it suffices for the verifier to calculate these indices. If the verifier finds that not all vertices are covered then the solution is discarded. The pseudocode for \emph{block verification} is mentioned in Algorithm \ref{verify}. 
Miners may collude and send out random solutions without doing any work. However, only a single honest
miner per computational task is enough to foil the entire
colluding effort. Colluding miners are not likely to cooperate especially if they know that all they need to do is a relatively small amount of useful work to win the competition
and claim the rewards. The rationale behavior will be to start
competing, and return the best solution. Thus, the mining game induces competition among the miners leading the system to act like a \emph{decentralized minimal dominating set solver}. 

\emph{Block Reward}: A miner who submits a dominating set of $G_T$ having least cardinality within $T^G_{max}$ wins the mining game and gets the block reward. This incentivizes rational miners to compete for finding the best solution. The block reward comprises fee from the transaction set in the block, and the fee provided by the utility company for solving the dominating set of the graph. The lookup table provides an estimate of the block interval time for a graph instance, indirectly providing some insight into the hardness of the problem. We expect a fair utility company to decide on the remuneration directly proportional to the hardness of the problem. In the next section, we discuss how the dominating set for $G$ is retrieved from the dominating set of $G_T$.
 
% One could question whether it be would more efficient to simply solve the problem directly, or release the problem in a marketplace to solve it, rather than trying to provoke competition among many miners. But this does not serve our purpose. Our motive is to combine the two individual problems in one system and realize a useful proof of work consensus. The dominating set problem ensures the security of the Bitcoin Blockchain. Since it is a permissionless network, any participant is free to join or leave anytime they want, and no single entity forces other participants to be a part of it.

\subsection{Retrieval of solution}
A miner finds the dominating set on the extended graph $G_T$ but not on graph $G$. But the utility company wants the dominating set on $G$. It would be wasting coin if the solution for $G_T$ cannot be mapped into a solution of $G$. We observe that $DS(G_T)$ is union of $V_{ds}$ and $W_{ds}$, where $V_{ds}$ are the vertices from $G$. The utility company checks if all vertices in $V_{ds}$ cover $G$. It could also be the case that not all vertices of $V_{ds}$ are required for covering $G$. Then those redundant vertices can be eliminated and the rest are added to $DS(G)$. If $G$ gets fully covered then $W_{ds}$ is discarded. If not all the vertices of $G$ are covered, then one needs to check which vertices in $W_{ds}$ are covering the remaining vertices in $G$. If a vertex $w_i$ is used for covering some vertices of $G$ then corresponding $v_i$ is added to the dominating set $DS(G)$. The reason behind this is that $w_i$ is connected to the neighbor of $v_i$. We provide the pseudocode in Algorithm \ref{ds-g} and justify in Lemma \ref{cc1} that the mapping results in a good solution for $G$ in expectation. 
 
\RestyleAlgo{ruled}
\begin{algorithm}[!ht]
			\caption{{\sf Find $DS(G)$}}
			\label{ds-g}
			%\fontsmall
   \textbf{Input}: Block $B$\\
   Parse block header of $B$ to get $DS(G_T),{id}_G, h_{\textrm{prev}\_\mathcal{B}},h_{MR},\delta$ and retrieve transaction set $\mathcal{T}$ from $B$\\
			Parse dominating set $DS({G_T})={V_{ds}}\cup {W_{ds}}$\\
   Set $visited\_set=\phi$\\
 Set $DS(G)=\phi$\\
                \For{$v$ in ${V}_{ds}$}
                {
               
                    Mark $v$ as visited \\
                  
                         Add $v$ to $visited\_set$\\
                         Add $v$ to $DS(G)$\\
                    
                    \For{$v' \in N(v)$}
                    {
                         If $v'$ is not visited then mark it visited \\  
                         Add $v'$ to $visited\_set$\\
                    }
                \If{$|visited\_set|=|V|$}
                {
                    return $DS(G)$
                }
                $S \leftarrow \mathcal{H}(h_{\textrm{prev}\_\mathcal{B}},h_{MR})$\\
               $Indices_S \leftarrow calc\_index(S)$\\
               \For{$v$ in $V\setminus visited\_set$}
                      {
                      \For{ $w \in Indices_S(N(v))$}
                    {
                        Mark $w$ as visited if not visited and add $w$ to $visited\_set$\\
                        $v'\leftarrow map\_index(w,G)$\\
                        Add $v'$ to $DS(G)$
                              %Stop scanning $M$\\      
                    } 
                    %   Start scanning $M$ from the index marking the neighbor set of $v$
                    % {
                    %     \If{bit is 1 for a vertex $v: v \in DS\_{G_T}$}
                    %     {
                                         
                    %     }
                    % } 
                      }
                    }
                
	\end{algorithm}

%\subsection{Solving Dominating Set Problem}
% \emph{Chrisimos} induces competition among miners to find the best possible dominating set in a given interval. A verifier selects the dominating set of lowest cardinality arriving within the block interval time, the rest of the solutions are discarded.

\subsection{Constructing the lookup table}
In hash-based PoW, a difficulty target is set and adjusted to stabilize latency between blocks \cite{antonopoulos2014mastering}. Robustness of the system is ensured as the computational power of the network varies with miners joining and leaving the network. However, in our protocol, real-life problem instances are submitted to be solved by the miners. The instances may significantly differ in the difficulties. In addition, the actual difficulty of each particular instance may not be known. Thus, it is necessary to estimate the time taken to generate and verify the block, which is the \emph{block interval time},  denoted as $T_{max}^G$.

\subsubsection*{Estimating Block Interval Time}
The block interval time comprises (a) block generation and (b) block verification. The crux of block generation is solving dominating set for the given graph instance. We describe each method and estimate the time taken.

% (a) \emph{Block Generation}: In Algorithm \ref{gen_block}, we have introduced the function \texttt{FindDominatingSet} over the extended graph $G_T$ without defining it. The miners are allowed to use an algorithm of their own choice. A miner needs to find a threshold on the size of the dominating set and based on that, plan its mining strategy. To find a dominating set of size less than $k$, a miner has to do an exhaustive search on the solution space for finding a dominating set in $G_T(V_T,E_T)$ of size $\frac{n'(1+\ln(1+\delta'))}{1+\delta'}$.
% \begin{theorem}
% \label{thresult}
% Finding a dominating set of size $k\leq \frac{n'(1+\ln(1+\delta'))}{1+\delta'}$ using exhaustive search for graph $G_T(V_T,E_T)$, where $n'=|V_T|$ and $\delta'$ is the minimum degree of the graph has time complexity $\mathcal{O}(e^{n'})$.
% \end{theorem}
% % \RestyleAlgo{ruled}
% % \begin{algorithm}[!htbp]
% % 			\caption{{\sf Greedy Heuristic for Dominating Set}}
% % 			\label{greedy}
% % \textbf{Input}: $G_T(V_T,E_T)$\\
% % %start\_time=current\_timestamp\\
% % Set $S_G=\phi$\\
% % Sort $V_T$ in descending order of degree\\

% % \While{$V_T\neq \phi$}
% % {
% %      Select vertex $v$ from $V_T$\\
% %       $S_G=S_G\cup \{v\}$\\
% %       \For{$w \in N(v)$}
%       {
%        $V_T=V_T\setminus \{w\}$
      
%       }
      
%       $V_T=V_T\setminus \{v\}$

% }
% %$\tau=current\_timestamp-start\_time$\\
% \Return{$S_G$}
%    \end{algorithm}

(a) \emph{Block Generation}: In Algorithm \ref{gen_block}, we have introduced the function \texttt{FindDominatingSet} over the extended graph $G_T$ without defining it. We analyze the runtime needed for doing an exhaustive search on the solution space for finding a dominating set in $G_T(V_T,E_T)$ of size $\frac{n'(1+\ln(1+\delta'))}{1+\delta'}$.
\begin{theorem}
\label{thresult}
Finding a dominating set of size $k\leq \frac{n'(1+\ln(1+\delta'))}{1+\delta'}$ using exhaustive search for graph $G_T(V_T,E_T)$, where $n'=|V_T|$ and $\delta'$ is the minimum degree of the graph has time complexity $\mathcal{O}(e^{n'})$.
\end{theorem}
\begin{proof}
The miners perform an exhaustive search to find the dominating set of size $k\leq \frac{n'(1+\ln (1+\delta'))}{1+\delta'}$, where $\delta'$ is the minimum degree of $G_T$. The following inequality holds from Sterling's second approximation:
$(\frac{n'}{k})^k \leq \binom{n'}{k} \leq (\frac{en'}{k})^k$

    %\item $x^{\frac{1}{x}}$ attains highest value when $x=e$.
When $k$ is the maximum value, the number of subsets of vertices of desired size will be $\binom{n'}{\frac{n'(1+\ln(1+\delta'))}{1+\delta'}}$ and thus we have:
    \begin{equation}
    \begin{split}
    (\frac{1+\delta'}{1+\ln(1+\delta')})^{\frac{n'(1+\ln(1+\delta'))}{1+\delta'}} \leq \binom{n'}{\frac{n'(1+\ln(1+\delta'))}{1+\delta'}} \leq \\ (\frac{e(1+\delta')}{1+\ln(1+\delta')})^{\frac{n'(1+\ln(1+\delta'))}{1+\delta'}}   
       \end{split}
    \end{equation}
 We assign $w=\frac{1+\delta'}{1+\ln(1+\delta')}$ which increases with $\delta'$ and the inequality becomes $w^{\frac{n'}{w}} \leq \binom{n'}{\frac{n'}{w}} \leq (ew)^{\frac{n'}{w}}$, where $(ew)^{\frac{1}{w}}$ attains highest value $e$ when $w=1$. Thus we have,
\begin{equation}
    \begin{split}
    \frac{1+\delta'}{1+\ln(1+\delta')}=1 \implies \delta'=\ln(1+\delta')
    \end{split}
\end{equation}
But this holds true if $\delta'=0$. For any graph, the minimum degree $\delta'\geq 1$ and $(ew)^{\frac{1}{w}}$ monotonically decreases as $w$ increases beyond 1. Thus for any $\delta'\geq 0$, $(ew)^{\frac{n'}{w}} \leq (e)^{n'}$. So the time complexity of the exhaustive search for a dominating set in a graph of order $n'$ is bounded by $e^{n'}$.
\end{proof}

It shows that searching exhaustively would result in an infeasible runtime for block generation. Another paper uses the \emph{Measure and Conquer} approach on NP-hard problem of finding minimum dominating set obtaining a tighter bound $\mathcal{O}(2^{0.598n'})$ \cite{fomin2009measure} but this is still exponential in the size of input. Thus, the miners are allowed to use an algorithm of their own choice that returns a result in polynomial time. To provide an estimate of the runtime for several synthetically generated datasets, we apply a greedy heuristic \cite{alon2016probabilistic} in the module \texttt{FindDominatingSet} of Algorithm \ref{gen_block}.
\RestyleAlgo{ruled}
\begin{algorithm}[!htbp]
			\caption{{\sf Greedy Heuristic for Dominating Set}}
			\label{greedy}
\textbf{Input}: $G_T(V_T,E_T)$\\
%start\_time=current\_timestamp\\
Set $S_G^*=\phi$\\
Sort $V_T$ in descending order of degree\\

\While{$V_T\neq \phi$}
{
     Select vertex $v$ from $V_T$\\
      $S_G=S_G\cup \{v\}$\\
      \For{$w \in N(v)$}
      {
       $V_T=V_T\setminus \{w\}$
      
      }
      
      $V_T=V_T\setminus \{v\}$

}
%$\tau=current\_timestamp-start\_time$\\
\Return{$S_G^*$}
   \end{algorithm}

The greedy heuristic works as follows: select the vertex with maximum degree $\gamma$, say $v_{\gamma}$, as the first element of $DS(G)$, and discard all the neighbors of $v_{\gamma}$ from $V_T$. Repeat the process by selecting the next vertex of maximum degree in $V_T\setminus DS(G)$ until $V_T=\phi$. The heuristic returns a dominating set $DS(G): |S_G^*|\leq |DS(G)| \leq \ln{\gamma} |S_G^*|$ where $S_G^*$ is the minimum dominating set.

\begin{theorem}
\label{th3}
The size of the dominating set fetched by the greedy heuristic is within the bound $\frac{n'(1+\ln(1+\delta'))}{1+\delta'}$ \cite{alon2016probabilistic}.    
\end{theorem}
\begin{proof}
We choose the vertices for the dominating set one by
one, when in each step a vertex that covers the maximum number of yet uncovered
vertices are picked.  If we pick a vertex $v \in V_T$, it will cover at least $\delta'+1$ vertices (including itself). Let's designate this set of vertices as $C(v)$. 

Suppose we select few such $v$ vertices sequentially and the vertices in the union of corresponding $C(v)$'s are covered. Let the number of vertices that do not lie in this union be $z$. Let's denote such a vertex with $u$. All such $u$ have their corresponding $C(u)\geq(\delta'+1)$. If we take the sum over all such $C(u)$, we get at least $z(\delta'+1)$. This will double count the connections of these $z$ vertices with the other vertices. But these connections can go to at most $n'$ vertices. Using averaging principle we can say that there exists one vertex which is included in at least $\frac{z(\delta'+1)}{n'}$ such $C(u)$ sets. If we select one such vertex out of the $z$ vertices, it will cover at least $\frac{z(\delta'+1)}{n'}$ vertices. The number of uncovered vertices will be at most $z\Big(1-\frac{\delta'+1}{n'}\Big)$. So we can say that at each step of selection, the number of uncovered vertices is reduced by a factor of $\Big(1-\frac{\delta'+1}{n}\Big)$. 

In the Greedy Heuristic, we start with $n'$ uncovered vertices, thus $z=n'$. After selection of $\frac{n'\ln{(\delta'+1)}}{\delta'+1}$ vertices sequentially, the number of uncovered vertices will be,
\begin{equation}
   n'\Big(1-\frac{\delta'+1}{n'}\Big)^{\frac{n'\ln{(\delta'+1)}}{\delta'+1}} < n'e^{-\ln{(\delta'+1)}} =\frac{n'}{1+\delta'}
\end{equation}
Thus after having $k=\frac{n'\ln{(\delta'+1)}}{\delta'+1}$ in the dominating set, we are still yet to select $\frac{n'}{\delta'+1}$ remaining uncovered vertices and form the dominating set. The size of the dominating set will be at most $\frac{n'(1+\ln(1+\delta'))}{1+\delta'}$. 

\end{proof}

This bound is loose and the result returned by the greedy heuristic has a cardinality lower than this. 

(b) \emph{Block Verification}: The verifier checks if the dominating set returned by (a) is valid using Algorithm \ref{verify}.

We use certain benchmark graph datasets and record the time taken to solve the minimal dominating set using the steps stated above in (a) and (b). The block interval time is pre-computed and maintained in a look-up table for all these graph instances. Time taken for solving the dominating set of the extended graph $G_T$ involves extending $G$ to $G_T$, which has a complexity of $\mathcal{O}(|E|)$. The time taken to find the dominating set using greedy heuristic is $\mathcal{O}(|V|)$ and the time taken for verification is again $\mathcal{O}(|V|)$. 
\begin{lemma}
 The time taken for the block generation is $\mathcal{O}(|E|)$   
\end{lemma}
\begin{proof}
The miner needs to extend the graph $G$ to $G_T$. By analyzing \emph{Extend} function, we observe that the output of $K(h_1,h_2)$ is $2|E|$. The miner needs to read this bit string and connect vertices in $G$ with vertices in $G_T\setminus G$. This is $\mathcal{O}(|E|)$. In the next step, to insert edges with probability $\frac{\hat{\delta}}{n}$ in $G_T\setminus G$, the miner keeps tracks of $\hat{\delta} \frac{n-1}{2}$ indices. This is $\mathcal{O}(|V|)$. Since $|V|<<|E|$ so the time complexity for constructing $G_T$ is $\mathcal{O}(|E|)$.

The miner now finds a dominating set in a graph instance $G_T(V_T,E_T)$ where $|V|=n, E \subseteq V \times V$ using greedy heuristic, defined in Algorithm \ref{greedy}. In each iteration, we pick up the vertex with the maximum degree, check its neighbor, add the vertex to the dominating set and delete it from the graph, and mark its neighbor as visited. Next, we select another unvisited vertex having the highest degree in the residual graph. This continues till no unvisited vertex remains in the graph. The number of iterations in the worst case is the sum of the degree of the vertices in the dominating set. But this just involves exploring all the vertices and thus it is $\mathcal{O}(|V|)$. This is again dominated by the time complexity of extending the graph. Thus the time complexity of Algorithm \ref{gen_block} is $\mathcal{O}(|E|)$.
\end{proof}
\begin{lemma}
 The time taken for the block verification is $\mathcal{O}(|V|)$.
 \end{lemma}
 \begin{proof}
 Upon analyzing Algorithm \ref{verify}, the verifier just needs to check the neighbors of the vertices in dominating set. It need not fully generate the graph $G_T$, just the information regarding connections of vertices in dominating set is required. The time complexity is the summation of the degrees of vertices in dominating set, without any double counting of a vertex. Thus, the time complexity is $\mathcal{O}(|V|)$.  
 \end{proof}

Thus, for estimating the time taken for adding the block, we take the product of all the components and figure out the scaling factor. We describe the procedure to estimate the block interval time for a new graph instance from the lookup table:\\
(A) We discuss the mining strategy followed by each miner: (i) Miner runs the greedy heuristic for dominating set problem on the extended graph $G_T$, and gets a solution of cardinality $k$. The problem miner tries to solve after the previous step: \emph{Does there exist a dominating set for $G_T$ of size less than $k$?} (ii) Given that the miner has sufficient time to report the solution, it applies all possible methods to obtain a solution having cardinality as low as possible. 
 
Thus, if the greedy heuristic on a benchmark instance takes a run-time $\tau$ units and the size of dominating set is $k': k'\leq k$, the miner can possibly find the dominating set of size $k'$ within time $t: 0<t< \tau$ using another efficient algorithm. If more time is provided, a miner may return a dominating set of cardinality less than $k'$. We set an interval of $l\tau : l \in \mathbb{R}^{+}, l>1$ so that both miner and verifier get enough time to propose and reach a consensus on the new block to be added. \\
(B) When a graph instance $G''(V'',E'')$ with minimum degree $\delta''$ is provided to the network, we check if the vertex count $|V''|$ of the instance matches with an entry of the lookup table. If entry $G'(V',E')$, minimum degree $\delta'$, in the lookup table where $|V'|=|V''|$, then that entry is selected. The edge count $|E'|$ of the entry is recorded and the time taken for proposing a new block for the given graph having edge count $|E''|$ will scale up (or scale down) by factor $\frac{(2|E''|+\delta''(|V''|-1)) \times |V''|}{(2|E'|+\delta'(|V'|-1)) \times |V'|}=\frac{2|E''|+\delta''(|V''|-1)}{2|E'|+\delta'(|V'|-1)}$. The edge count is considered with respect to the extended graphs of both $G'$ and $G''$. \\
(C) If there is no entry that matches the vertex count $|V''|$ then the entry with the maximum vertex count less than $|V''|$ is selected. Let that instance be $G^*(V^*,E^*)$, with minimum degree $\delta^*$, where vertex count of the entry is $|V^*|$ and edge count of the entry is $|E^*|$. Time taken for proposing a new block for the given graph having edge count $|E''|$ will scale up (or scale down) by factor $\frac{(2|E''|+\delta''(|V''|-1)) \times |V''|}{(2|E^*|+\delta^*(|V^*|-1)) \times |V^*|}$. 

\begin{figure}[!htp]
% \centering
\includegraphics[scale=0.35]{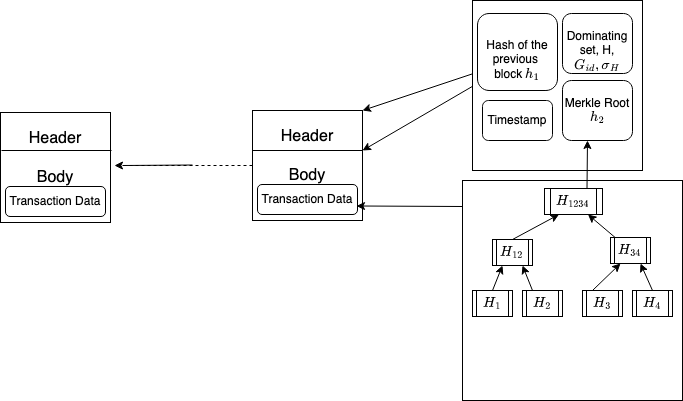}
\caption{Block Structure when \emph{Chrisimos} is used for generating Blockchain}
\label{fig:block}
\end{figure}

The structure of a block added to the blockchain is shown in Fig. \ref{fig:block}. It is similar to the one used for hash-based PoW except for the nonce part. Instead of the nonce, it now contains $id$ of the graph instance, the hash of the graph instance and a valid signature, the minimum degree of the graph, and the dominating set of the extended graph. 

\section{Correctness of Chrisimos}
\label{analysis}
% \label{5}
\begin{lemma}
\label{th4}
The expected number of vertices of $G_T$ covered by the dominating set $DS(G)$ is $\frac{3n}{2}$.
\end{lemma}

\begin{proof}
The given graph $G(V,E)$ has $n$ vertices namely $v_1,v_2,\ldots,v_n$ and we expand $G$ to $G_T(V_T,E_T)$ using our construction technique where we add $n$ new vertices namely $w_1,w_2,\ldots,w_n$. As per our construction, $w_1$ is connected to each neighbor of $v_1$ with probability $\frac{1}{2}$. Now if a vertex $v_i$ has an edge with $v_j$ in $G$, then $w_i$ has an edge with $v_j$ in $G_T$ with probability $\frac{1}{2}$. So if a vertex $v_i$ is connected to $d_i$ number of vertices of $V$ then the expected number of neighbors of $v_i$ in $V_T \setminus V$ will be $\frac{d_i}{2}$. Using the linearity property of expectation it can be said that the number of vertices in $V_T \setminus V$ covered by a dominating set $DS(G)$ of graph $G$ will be $\frac{n}{2}$. So the expected number of vertices of $G_T$ covered by $DS(G)$ will be $(n+\frac{n}{2})=\frac{3n}{2}$.  
\end{proof}

 The expected number of vertices in $V_T \setminus V$ not yet covered by $DS(G)$ is $\frac{n}{2}$. The next theorem gives an upper bound on the size of the dominating set for these $\frac{n}{2}$ vertices. 
\begin{lemma}
\label{th5}
If the probability of an edge in the induced sub-graph of $G_T$ on all the vertices $w_i$ for $i \in \{1,2,\ldots,n\}$ is $p$ then the expected size of the dominating set of these $\frac{n}{2}$ vertices in this induced graph will be $b$ where $b$ is the minimum integer for which $\frac{n}{2}(1-p)^b<1$.
\end{lemma}

\begin{proof}
If we select a random vertex among this $\frac{n}{2}$ vertices for inclusion in its dominating set, then the expected number of vertices it will cover is $(\frac{n}{2}-1)p+1 \approx \frac{np}{2}$. So the remaining vertices to cover will be $(\frac{n}{2}-\frac{np}{2})=\frac{n}{2}(1-p)$. So selecting a vertex in the dominating set will reduce the size of the set of vertices to cover by a factor of $(1-p)$. Similarly after selecting $b$ such vertices in the dominating set the remaining number of vertices to cover will be $\frac{n}{2}(1-p)^b$. When $\frac{n}{2}(1-p)^b<1$ then there will be no vertex left to cover. From this inequality, we get the minimum $b$ which is also the cardinality of the dominating set of these $\frac{n}{2}$ vertices. 
\end{proof}

\begin{lemma}
\label{c1}
If $DS(G)$ covers $\frac{3n}{2}$ vertices of $G_T$, where $|DS(G)|\leq \frac{n(1+\ln(1+\delta))}{1+\delta}$, and remaining $\frac{n}{2}$ vertices in $G_T-G$ is covered by a dominating set of size $b: \frac{n}{2}(1-p)^b<1$, then the probability of an edge connecting vertices in $G_T-G$ must be greater than $\frac{\delta}{n}$.
\end{lemma}

\begin{proof}
We have $DS(G): |DS(G)|\leq \frac{n(1+\ln(1+\delta))}{1+\delta}$, that covers $\frac{3n}{2}$ vertices of $G_T$. Given the expected size of the dominating set of the remaining $\frac{n}{2}$ vertices of $G_T$ is $b$ and $DS(G_T)$ is the dominating set of $G_T$, $|DS(G_T)|$ must be less than $(|DS(G)|+b)$ for getting the dominating set within the bound stated in Theorem \ref{th1}. Therefore, $|DS(G_T)|\leq |DS(G)|+b \implies b \geq |DS(G_T)|-|DS(G)|$.

%Let $G(V,E)$ be the graph of our interest. $|V|=n, |E|=e$. We want a dominating set $D$ of $G$.
To construct $G_T$, we extend $G$ to $G_T(V_T,E_T)$ where $V_T=2n$ and $V\subset V_T, E\subset E_T$. For edges in $G_T-G$, an edge exists between a given pair of vertices $w_i$ and $w_j: w_i,w_j \in V_T\setminus V$ with probability $p$. It is observed that due to the extension procedure, the expected degree of a vertex $v_i \in G$ increases from $d_i$ to $\frac{3d_i}{2}$ and $DS(G)$ covers $\frac{n}{2}$ new vertices of $G_T$. Also, the expected degree of a vertex $w_i \in V_T\setminus V$ becomes $\frac{d_i}{2}+p(n-1)$. The size of our desired dominating set depends on the minimum degree of the graph. Finding the dominating set of the graph $G$ is our main goal and to integrate it into the Blockchain structure we extend $G$ to $G_T$. We expect all the important factors of our system to be components of $G$ and thus we set $p$ so that the vertex with a minimum degree in $G_T$ comes from set $V$. 
%So it is expected that the following inequality holds true and gives us a choice for the value of $p$.
%Let $G(V,E)$ be the graph of our interest. $|V|=n, |E|=e$. We want a dominating set $D$ of $G$. First we extend it to $G_T(V_T,E_T)$ where $V_T=2n$ and $V\subset V_T, E\subset E_T$. For edges in $G_T$, we connect the new $n$ vertices with the old $n$ vertices using the method described in the previous section. For edges between new vertices, we will fix a probability $p$ for edges. The value of $p$ should be specified accordingly later. It is observed that due to the extension procedure, the expected degree of a vertex $v_i \in G$ increases from $d_i$ to $\frac{3d_i}{2}$ and $D$ covers $\frac{n}{2}$ new vertices of $G_T$. Also the expected degree of a vertex $w_i \in V_T\setminus V$ becomes $\frac{d_i}{2}+p(n-1)$. To make the constraint on the size of the dominating set dependent on some vertex $v_i$ instead of some dummy vertex $w_i$ the vertex with minimum degree in $G_T$ must belong from $V$. So,
The following inequality $\frac{\delta}{2}+p(n-1) \geq \frac{3}{2}\delta \implies p \geq \frac{\delta}{n-1} \approx \frac{\delta}{n} (\because n>>\delta)$ gives the lower bound on $p$. This $p$ decides the probability with which vertices in set $V_T\setminus V$ must be interconnected. The miners will find a dominating set $DS(G_T)$ of $G_T$ where $|DS(G_T)| \leq \frac{2n(1+\ln(1+\frac{3\delta}{2}))}{1+\frac{3\delta}{2}}$. Additionally, we need a dominating set $DS(G)$ of $G$ where $|DS(G)| \leq \frac{n(1+ln(1+\delta))}{1+\delta}$. From Lemma \ref{th5}, we have,
% So, we need to get $D$ from $D_T$. Now if the size of the dominating set of the induced subgraph of $G_T$ on $\frac{n}{2}$ uncovered vertices by $D$ is greater than $D_T - D$ then we can get $D$ of our desired size from $D_T$. This sub graph has $\frac{n}{2}$ total vertices with probability of each edge as $p$. While we select an vertex in the dominating set for this subgraph, we cover $p$ times the total number of vertices to cover till that point. So, $(1-p)$ part of the vertex set remains to cover. So, it can be said that if $k$ is the size if the minimum dominating set of this graph then,
\begin{equation}
\begin{split}
    \frac{n}{2}(1-p)^b<1 \implies (1-p)^b<\frac{2}{n}  \implies b\ln{(1-p)}<\ln{\frac{2}{n}}   \implies b\ln{\frac{1}{1-p}}>\ln{\frac{n}{2}}
  \implies \\ b>\frac{\ln{\frac{n}{2}}}{\ln{\frac{1}{1-p}}} = \frac{\ln{\frac{n}{2}}}{\ln{(1+p+p^2+\dots)}} > \frac{\ln{\frac{n}{2}}}{p+p^2+p^3+\dots} = \frac{\ln{\frac{n}{2}}}{p(1+p+p^2+\dots)} 
  \end{split}
\end{equation} 
So, $|DS(G_T)|-|DS(G)| \leq \frac{2n(1+\ln(1+\frac{3\delta}{2}))}{1+\frac{3\delta}{2}}-\frac{n(1+ln(1+\delta))}{1+\delta} \leq b$. Substituting $b=\frac{\ln{\frac{n}{2}}}{p(1+p+p^2+\dots)}$ in the above equation, it is sufficient to show,
\begin{equation}
\label{eql}
  \frac{\ln{\frac{n}{2}}}{p(1+p+p^2+\dots)}    \geq \frac{2n(1+\ln(1+\frac{3\delta}{2}))}{1+\frac{3\delta}{2}}-\frac{n(1+ln(1+\delta))}{1+\delta} 
\end{equation}

The L.H.S of the Eq. \ref{eql} decreases with an increment in $p$. We need to check whether the above inequality holds upon substituting $p=\frac{\delta}{n}$. 

$\forall x\geq 1, \frac{ 1+\ln(1+x)}{1+x}<\frac{1+\ln x}{x}$, so we have,
\begin{equation}
\begin{split}
     \tfrac{(1-p)\ln{\tfrac{n}{2}}}{p}\geq (\tfrac{n}{\delta}-1)(\ln{\tfrac{n}{2}}) > \frac{n(1+\ln(\delta))}{\delta}>  \frac{2n(1+\ln(1+\delta))}{1+\delta}-\\\frac{n(1+ln(1+\delta))}{1+\delta}> \frac{2n(1+\ln(1+\frac{3\delta}{2}))}{1+\frac{3\delta}{2}}-\frac{n(1+ln(1+\delta))}{1+\delta}
\end{split}
\end{equation}

%   \end{split}
%   \end{equation}
% %   \begin{equation}
%   \begin{split}
%    % \implies  (\tfrac{n}{\delta}-1)(\ln{\tfrac{n}{2}})> \frac{4n(1+\ln(\frac{3\delta}{2}))}{3\delta}-\frac{n(1+ln(\delta))}{1+\delta}\\
%      > >\frac{n(1+\ln(1+\delta))}{1+\delta}\\
% \end{split}
% \end{equation}
% Since $\delta \geq 1$, $2\delta\geq \delta+1$, so we can write,
% \begin{equation}
% \begin{split}
%       (\tfrac{n}{\delta}-1)(\ln{\tfrac{n}{2}})> \frac{nln \delta}{1+\delta} \geq \frac{nln \delta}{2\delta}>(\frac{n}{\delta}-1)\frac{ln \delta}{2}
% \end{split}
% \end{equation}
 which is true as $\frac{n}{\delta}-1\approx \frac{n}{\delta}$ and $\delta <\frac{n}{2e}$. It is fair enough to consider $\delta < \frac{n}{2e}$. If the minimum degree is greater than this, then it becomes easier to find dominating sets for graphs and that would not be provided as input instances to the Blockchain for mining.
\end{proof}

% \subsection{Correctness of Output}
% \label{correct}
% \begin{lemma}
% \label{cf1}
% The dominating set for $G_T(V_T,E_T)$ is within the bound stated in Theorem \ref{th1}.   
% \end{lemma}

% \begin{proof}
%     To verify the PoW solutions one peer need to check whether the solution provided is the dominating set of $G_T$. The peers don’t need to be sure that the dominating set found by a miner is the minimum dominating set in the graph. She can only choose the smallest one she has received in the epoch and can accept the corresponding block. An honest miner will check whether the dominating set for $G_T$ follows the bound mentioned in Theorem \ref{th1}. Additionally, a rational miner will use her maximum computation power to ensure that the solution is not just satisfying the bound but has a low cardinality. It is highly unlikely that miners collude and submit a bad solution (i.e. exceeding the bound). Even if there exists one miner that provides a better solution than the rest, he gets rewarded. 
% \end{proof}
\begin{lemma}
\label{cf1}
The dominating set for $G_T$ is within the bound stated in Theorem \ref{th1}.   
\end{lemma}
\begin{proof}
    To verify the PoW solutions one peer need to check whether the solution provided is the dominating set of $G_T$. The peers don’t need to be sure that the dominating set found by a miner is the minimum dominating set in the graph. She can only choose the smallest one she has received in the epoch and can accept the corresponding block. An honest miner will check whether the dominating set for $G_T$ follows the bound mentioned in Theorem \ref{th1}. Additionally, a rational miner will use her maximum computation power to ensure that the solution is not just satisfying the bound but has a low cardinality. It is highly unlikely that miners collude and submit a bad solution (i.e. exceeding the bound). Even if there exists one miner that provides a better solution than the rest, he gets rewarded. 
\end{proof}

\begin{lemma}
\label{cc1}
Given $|DS(G_T)|$ is within the bound stated in Theorem \ref{th1} then in expectation, $|DS(G)|$ follows this bound.
\end{lemma}

\begin{proof}
    When $G$ is extended to $G_T$, the expected degree of each vertex in $G$ increases by a factor $\frac{1}{2}$. The expected minimum degree of $G_T$ increases to $\frac{3\delta}{2}$. If a miner starts with finding the dominating set for $G$, then as per Lemma \ref{th4}, $\frac{n}{2}$ vertices of $G_T\setminus G$ are already covered in expectation. From Lemma \ref{th5} and Lemma \ref{c1}, we infer that upon setting the probability of edge formation in $G_T\setminus G$ to $\frac{2\delta}{n}$, we get a dominating set of cardinality at least $\frac{2n(1+\ln(1+1.5\delta)}{1+1.5\delta}-\frac{n(1+\ln(1+\delta)}{1+\delta}$, covering the remaining $\frac{n}{2}$ vertices of $G_T\setminus G$. If the expected upper bound on $|DS(G_T)|$ is $\frac{2n(1+\ln(1+1.5\delta)}{1+1.5\delta}$ as per Theorem \ref{th1}, then $|DS(G)| \leq \frac{n(1+\ln(1+\delta)}{1+\delta}$ in expectation.

   A miner can use the strategy of finding the dominating set of $G_T\setminus G$ and then find the dominating set for $G$. But there is a risk that the cardinality may exceed $\frac{2n(1+\ln(1+1.5\delta)}{1+1.5\delta}-\frac{n(1+\ln(1+\delta)}{1+\delta}$. But using the strategy of finding the dominating set of $G_T$ within the bound ensures the dominating set of $G$ has a cardinality lower than $\frac{n(1+\ln(1+\delta)}{1+\delta}$. The above analysis proves that any strategy applied by the miner to find the dominating set for $G_T$ does not degrade the solution quality of the dominating set for graph $G$ in expectation. 
\end{proof}

\section{Security analysis}
\label{security}
% Our goal is to prove that the Nakamoto consensus using our proposed PoW mechanism guarantees \emph{safety} and \emph{liveness} \cite{ren2019analysis}. We define the properties as follows:\\
% (i) \emph{Safety}: Honest miners do not commit different blocks at the same height.\\
% (ii) \emph{Liveness}: If all honest miners in the system attempt to include a certain input block then, after a few rounds, all miners report the input block as stable.
Our goal is to prove that the Nakamoto consensus using our proposed PoW mechanism guarantees \emph{safety} and \emph{liveness} \cite{ren2019analysis}. We define the properties as follows:\\
(i) \emph{Safety}: Honest miners do not commit different blocks at the same height.\\
(ii) \emph{Liveness}: If all honest miners in the system attempt to include a certain input block then, after a few rounds, all miners report the input block as stable.

It may happen that two or more miners solved their instances at about the same time and published their blocks, creating the situation that is known as a fork in blockchain systems. Forks are usually resolved in the synchronization phase using the rule specified for the particular Blockchain. Only one of the blocks will pass both the verification and synchronization phases. We define a chain selection rule whereby the chain having maximum work done is selected.

\subsection{Chain selection rule in \emph{Chrisimos}}
\label{chain}
%\pb{is this worth a separate section?}
Work done in a chain is summation of the work done in the individual blocks forming the chain. The work done in a block is proportional to the size of the graph as well as the effort put by the miner in finding a dominating set of lower cardinality. Thus, we define the work done as the product of the number of edges and the number of vertices in the extended graph, scaled by the ratio of the expected permissible size of dominating set for the extended graph and the size of the dominating set returned as a result. The work done in block $B$ is defined as:
\begin{equation}
    \textrm{Work-done}_B=\Big( 2|E|+\frac{\hat{\delta}(|V|-1)}{2} \Big) \times |V| \times \frac{\frac{2n(1+\ln(1+\frac{3\delta}{2}))}{1+\frac{3\delta}{2}}}{|DS(G_T)|}
\end{equation}
  
where $|E_T|=\Big( 2|E|+\frac{\hat{\delta}(|V|-1)}{2} \Big)$ and expected permissible size $DS(G_T)\leq \frac{2n(1+\ln(1+\frac{3\delta}{2}))}{1+\frac{3\delta}{2}}$ (minimum degree of $G_T$ increases to $\frac{3\delta}{2}$ in expectation). If there is a fork at block $B'$, then the miner chooses the chain, starting from $B'$, having the highest work done, i.e. if $\mathcal{C}=\{C_1,C_2,\ldots,C_m\}$ be $m$ such forks from block $B'$, then choose the chain $C_i \in \mathcal{C}$ such that $\sum_{B \in C_i: C_i \in \mathcal{C}} \textrm{Work-done}_B$ is maximum. An honest miner considers a block $B$ committed if $B$ is buried at least $f$ blocks deep in its adopted chain. We do not quantify the value of $f$ as it must be correlated to the winning chance of an adversary to mine a valid block and we keep the analysis as part of future work.

We summarize the rules for checking the validity of a chain as follows:\\
 (i) Verifier rejects a block that has a graph instance with id either less than or equal to the graph id of the instance mined in the previous block of the main chain. If the graph instance has a malformed signature (not signed by the committtee) or the hash of the graph does not match, verifier rejects the block. \\
 (ii) \emph{A block is said to be checkpointed if it has received at least $f$ confirmations}: All blocks before the last checkpointed block is also considered checkpointed. Additionally, we assume all honest miners reach a consensus over a single chain having the last checkpointed block. This implies that within the next $f$ blocks, the view of the network will be the same for all the miners till the last checkpointed block but can differ after that.\\
 (iii) \emph{If the miners observe another sub-chain with higher work done}: In this situation majority of the honest miners have a consensus on the view of the main chain till the last checkpointed block (i.e. the block that has $f$ confirmations). Now the miners encounter another sub-chain that has induced a fork. The following cases can happen:
 \begin{itemize}
 \item If any block in the other sub-chain violates rule (i), discard the chain.
     \item If the fork starts from a block before the last checkpointed block, then the chain is discarded. 
 \item If the fork starts from or after the last checkpointed block, then there could be two possible cases :\\
 (a) If the fork occurs from the checkpointed block, then all the sub-chains having a length less than $f$ will be discarded. For the rest of the sub-chains with a length of at least $f$, follow the rule (b) mentioned below.\\
 (b) If fork occurs after the checkpointed block, then select the sub-chain with the highest work done. If there is a tie, randomly choose one of the chains.
 \end{itemize}

% The security of our scheme in terms of double-spending and selfish mining attacks  is ensured by this chain selection rule. We do not quantify the value of $f$ as that would be correlated to the winning chance of an adversary to mine a valid block. The analysis would require a different mathematical model and we keep it as part of future work. 
%\iffalse
\subsection{Security analysis}
We show that in the proposed scheme, the estimated time for block addition is sufficient for a graph instance to be solved guaranteeing progress. Since all the graph instance is solved and added sequentially into the blockchain, we argue the safety property in terms of selfish-mining attacks and double-spending attacks \cite{nicolas2019comprehensive}. The property of safety and liveness holds in the synchronous model. In an asynchronous setting, it is tricky to argue whether all graph instances announced in the network get added to the blocks of the blockchain. We leave this analysis as a part of the future work. We state certain lemmas that justify the safety and liveness of \emph{Chrisimos}. An informal proof of these lemmas has been discussed in Section \ref{security}. 

\begin{lemma}
\label{x}
    For a graph instance $G$, the block time interval $T_{max}^G$ is sufficient for adding the block to the Blockchain.
\end{lemma}
\begin{proof}
We had shown during the estimation of block interval time that it takes into account the block generation as well as verification time. If the block generation time is $\tau$ (if the graph instance is already present in the lookup table), we set $T_{max}^G$ to $l\tau: l>1$. If the graph instance is not present in the lookup table, the time is estimated by scaling it based on the edge count and vertex count of the graph instance. Since the lookup table is prepared by using a greedy heuristic, a rational miner will definitely get a solution by at least using the greedy heuristic. 
\end{proof}
%The proof of these lemmas follows from the chain selection rule. 
%When a miner adopts a new highest-work-done chain, either through mining or by receiving from other miners, it broadcasts and mines on top of the highest-work-done chain. 

% \begin{lemma}
% \label{x}
%     For a graph instance $G$, the block time interval $T_{max}^G$ is sufficient for adding the block to the Blockchain.
% \end{lemma}
\iffalse
We had shown during the estimation of block interval time that it takes into account the block generation as well as verification time. If the block generation time is $\tau$ (if the graph instance is already present in the lookup table), we set $T_{max}^G$ to $l\tau: l>1$. If the graph instance is not present in the lookup table, the time is estimated by scaling it based on the edge count and vertex count of the graph instance. Since the lookup table is prepared by using a greedy heuristic, a rational miner will definitely get a solution by at least using the greedy heuristic. 
\fi

\begin{lemma}
The probability of a double spending attack on any block before the last checkpointed block is negligible.    
\end{lemma}

This follows from rule (iii) of \emph{chain selection rule} where the honest majority has consensus till the last checkpointed block. Any fork after this will lead to the selection of the sub-chain having the highest work done.

\begin{lemma}
The probability of a selfish mining attack is negligible provided the signature scheme used by committee members is universally unforgeable and the majority of miners in the blockchain network are honest.
\end{lemma}
\begin{proof}

We provide a proof sketch for the following lemma based on the chain rule defined in Section \ref{chain}. The adversary has to start mining the private chain after the last checkpointed block else as per rule (iii), its privately mined chain will anyway get discarded. We do consider the two cases where the miner will start privately mining from the checkpointed block or after that:\\
(a) If the malicious miner induces a fork from the last checkpointed block $B_l$, the malicious private subchain must be of length $\geq f$ as per rule (iii)(a) mentioned in Section \ref{chain}. The chain of the adversary remains the same till block $B_l$ and starts differing from here, so we label these blocks as $B_{l+i}', 1\leq i\leq z$ where $z\geq f$. From Lemma \ref{x}, any block $B_{l+t}', 1\leq t \leq f$ mined by the adversary must have fetched the dominating set on the same graph instance as that in block $B_{l+t}$. For the subchain having $B_{l+i}'$ blocks to be selected, $\sum_{i=1}^z \textrm{Work-done}_{B_{l+i}'}>\sum_{i=1}^f \textrm{Work-done}_{B_{l+i}}$. From Lemma \ref{x}, $z$ cannot exceed $f+1$. This is because each graph instance is provided sequentially after elapse of block interval time. If the malicious miner luckily finds a solution quite earlier than the rest of the miners for the $(l+f+1)^{th}$ instance, it will have a higher chance of winning due to the impact of an additional block $B_{l+f+1}$. If $z=f$, then the cumulative work done in all the blocks of the sub-chain must be higher. This would require the adversary to be lucky in at least one of the blocks where it had managed to find a better solution and the rest of the blocks in the sub-chain must provide a solution as good as the one provided by an honest miner. \\
(b) If the adversary starts selfish mining after the checkpointed block then any sub-chain of length less than $f$ would do but it needs to follow the criteria (iii) (b) of the \emph{chain selection rule} to win the mining game. 

In both cases, if the adversary finds that at $i^{th}$ block, $1\leq i \leq f$, the cumulative work done in his private sub-chain is less than the cumulative work done in the main chain then it is highly likely he will abandon selfish mining and try to add the new block on the main chain. To continue mining on his sub-chain, he has to calculate the expected probability that he wins the mining game by adding the next block and this is conditioned on the fact that others must return a result far worse than what the adversary does.

Another possibility to launch selfish mining is by privately mining a longer chain. It is not possible for a miner to generate several graph instances and mine a longer chain. It follows from rule (i) of \emph{chain selection rule}, where a verifier will reject any illegitimate graph instance not signed by the committee members. Since we assume that the majority of the committee members are honest, and a universally unforgeable signature scheme is used to sign the instance, miners will not be able to forge signatures for all the instances. Hence, the adversary will able to pull off the attack with negligible probability.
\end{proof}
\section{Experimental Analysis}
\label{experiment}
% All the miner needs is the capability to store the graph, do the transformation and then find the dominating set. After a graph instance is part of the chain, a miner can just delete the graph. 

\textbf{Setup}: For our experiments, we use Python 3.10.0, and NetworkX, version 3.1 - a Python package for analyzing complex networks. System configuration used is Intel Core i5-8250U CPU, Operating System: macOS 12.4, Memory: 8 GB of RAM. Social networks and other utility networks follow power-law model where few vertices are central to the graph instance. Thus we use synthetically generated graph instances (based on Barab{\'a}si-Albert Model \cite{albert2002statistical} and Erdős-Rényi Model \cite{seshadhri2012community}) mimicking this model to generate the benchmark datasets. We choose appropriate parameters to generate the synthetic graph instances such that they simulate real-life networks. The order of the graph varies between 1000 to 200000, and the average degree of the graph varied between 10 to 50. For a given number of vertices, we took the average over all the instances with varied edge counts. We run Algorithm \ref{gen_block} to estimate the block proposer time. We use the greedy heuristic in the module \texttt{FindDominatingSet} but the miner is free to choose any algorithm. To estimate the time taken to verify the block, we run Algorithm \ref{verify}. 

\textbf{Observation}: The plot in Fig. \ref{fig:plot} shows the impact of increasing the size of the graph instance (increase in number of vertices) on the block proposer time (or generation time) and verification time. For small-order graphs, it is of the order of seconds but for graphs having vertex count more than 75000, the block proposer time goes up to 7 mins. The time taken by the verifier shows a slow and steady increase and it is less than a minute even for a graph of size 200000. In our experiment, we extend the graph using the \emph{Extend }function, and that results in a doubling of vertex count. So if we report the result for vertex count 200000, it actually denotes the execution time of finding and verifying dominating set on a graph of size 400000. The block generation time is approximately 5 times that of the block verification time.
\begin{figure}[!ht]
    \centering
    \includegraphics[scale=0.35]{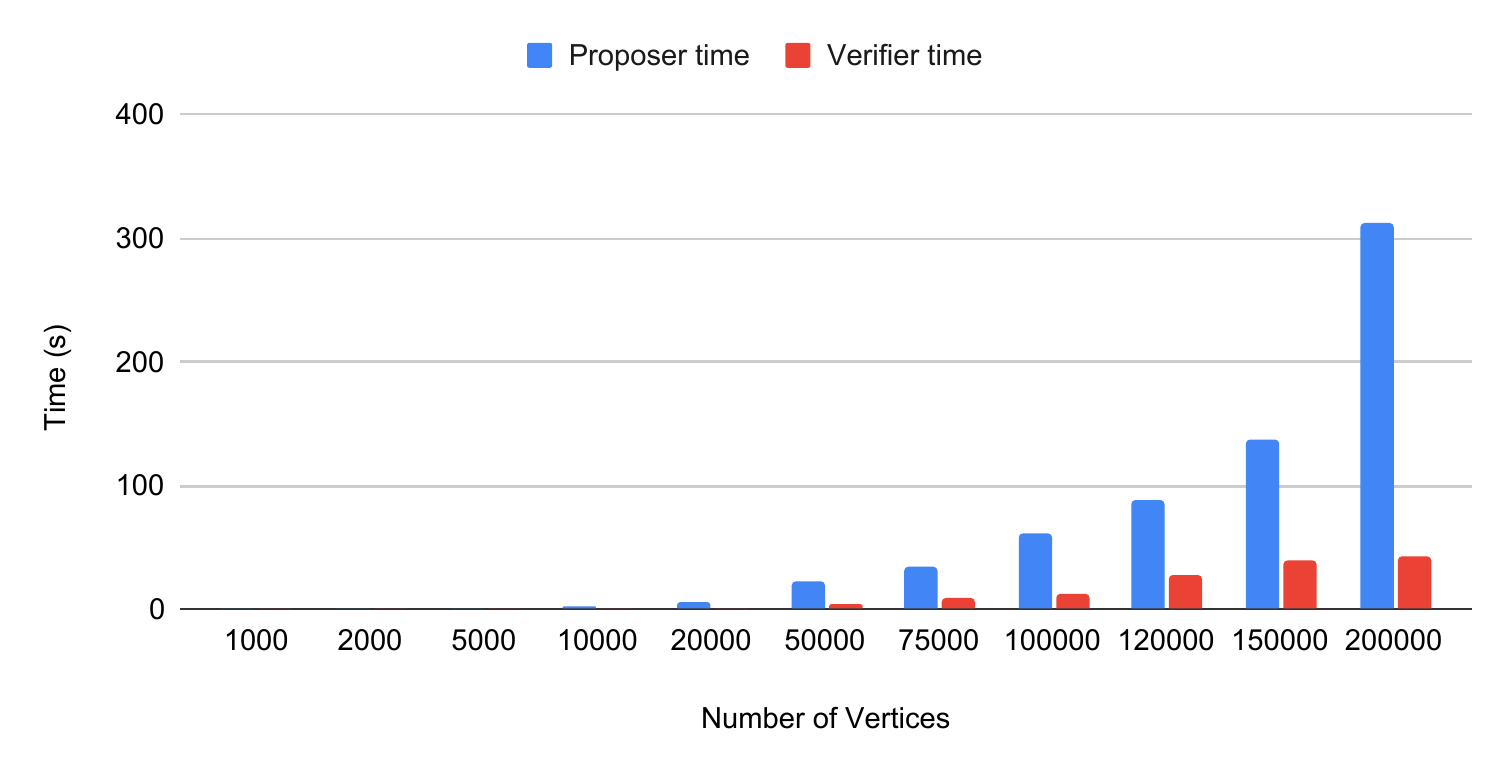}
    \caption{Plot showing impact of increasing number of vertices on block proposer time and block verification time}
    \label{fig:plot}
\end{figure}

\textbf{Discussion}: We observe that the block generation time increases almost linearly with an increase in the number of vertices in the graph instance. This is because the block generation time is dependent on exploring all the edges in the given graph instance. We consider moderately connected graph instances. Finding dominating set in too sparse or too dense graph becomes easier so no utility company would provide such graph instances. On the contrary, the verification time increases slowly (but linearly) compared to the block generation time with an increase in the size of the input graph. The reason is that the verifier just needs to check whether the vertices in the dominating set cover the entire graph. The time complexity is bounded by the number of vertices in the graph and it is less than the number of edges in the graph. This is fair enough as block verification must be done faster than block generation in any blockchain to allow more miners to join the Bitcoin network without too much spend on verification. 

Our results show that if a graph instance of an order as large as 200000 is provided for mining, the block interval time will be around 10 min to 12 min on average. This is similar to that of hash-based PoW. The block interval time might increase if the graph size increases but then the throughput of the Blockchain will reduce and hence a decision has to be made whether graph instance larger than 200000 must be allowed at the cost of reduced throughput.

\section{Conclusion}
\label{conclusion}
We propose \emph{Chrisimos}, a useful Proof-of-Work that solves a problem having real-life utility instead of wasting resources in generating nonce for hash-based PoW. Miners are asked to find a minimal dominating set on a real-life graph instance. Finding a minimum dominating set of size less than a positive integer is an NP-complete problem so miners use a heuristic of their choice. Miners return a solution and verifiers collect all the solutions that arrive within a given block interval time. Verifier selects the solution with the lowest cardinality and this simulates a decentralized minimal dominating set solver. We also mention a new chain rule that resolves disputes in the event of a fork and ensures security of our proposed PoUW. 
%Through experimental analysis and formal proof, we show that our scheme is secure and resolves problem of wastage of energy faced for hash-based PoW by replacing it with a graph theoretic problem.   

As a part of future work, we will propose a mathematical model that would allow us to quantify the computation power of the adversary and figure out its chance of winning the mining game over an honest miner. Based on the mathematical model, we will analyze the various attacks observed on the Bitcoin blockchain. We would also propose a generic model that would allow miners to accept any NP-complete problem and provide a solution to that. Additionally, we want to utilize pool mining to scale the system where each miner in the pool would solve the problem partly and then aggregate the partial solutions in generating the final solution.
 
\section*{Acknowledgment}
We thank Department of Computer Science and Information Systems, BITS Pilani, KK Birla Goa Campus, Goa, India for funding our research.
%\bibliographystyle{plainurl}
%\section*{Acknowledgment}
%\pb{Few of the references formatting needs to be corrected.}
\bibliographystyle{ACM-Reference-Format}
\bibliography{trustcom}


\begin{thebibliography}{37}

%%% the Web address.

\ifx \showCODEN    \undefined \def \showCODEN     #1{\unskip}     \fi
\ifx \showDOI      \undefined \def \showDOI       #1{#1}\fi
\ifx \showISBNx    \undefined \def \showISBNx     #1{\unskip}     \fi
\ifx \showISBNxiii \undefined \def \showISBNxiii  #1{\unskip}     \fi
\ifx \showISSN     \undefined \def \showISSN      #1{\unskip}     \fi
\ifx \showLCCN     \undefined \def \showLCCN      #1{\unskip}     \fi
\ifx \shownote     \undefined \def \shownote      #1{#1}          \fi
\ifx \showarticletitle \undefined \def \showarticletitle #1{#1}   \fi
\ifx \showURL      \undefined \def \showURL       {\relax}        \fi
% invisible to TeX
\providecommand\bibfield[2]{#2}
\providecommand\bibinfo[2]{#2}
\providecommand\natexlab[1]{#1}
\providecommand\showeprint[2][]{arXiv:#2}

\bibitem[Albert and Barab{\'a}si(2002)]%
        {albert2002statistical}
\bibfield{author}{\bibinfo{person}{R{\'e}ka Albert} {and}
  \bibinfo{person}{Albert-L{\'a}szl{\'o} Barab{\'a}si}.}
  \bibinfo{year}{2002}\natexlab{}.
\newblock \showarticletitle{Statistical mechanics of complex networks}.
\newblock \bibinfo{journal}{\emph{Reviews of modern physics}}
  \bibinfo{volume}{74}, \bibinfo{number}{1} (\bibinfo{year}{2002}),
  \bibinfo{pages}{47}.
\newblock


\bibitem[Alon and Spencer(2016)]%
        {alon2016probabilistic}
\bibfield{author}{\bibinfo{person}{Noga Alon} {and} \bibinfo{person}{Joel~H
  Spencer}.} \bibinfo{year}{2016}\natexlab{}.
\newblock \bibinfo{booktitle}{\emph{The probabilistic method}}.
\newblock \bibinfo{publisher}{John Wiley \& Sons}.
\newblock


\bibitem[Antonopoulos(2014)]%
        {antonopoulos2014mastering}
\bibfield{author}{\bibinfo{person}{Andreas~M Antonopoulos}.}
  \bibinfo{year}{2014}\natexlab{}.
\newblock \bibinfo{booktitle}{\emph{Mastering Bitcoin: unlocking digital
  cryptocurrencies}}.
\newblock \bibinfo{publisher}{" O'Reilly Media, Inc."}.
\newblock


\bibitem[Asgarnezhad and Torkestani(2011)]%
        {asgarnezhad2011connected}
\bibfield{author}{\bibinfo{person}{Razieh Asgarnezhad} {and}
  \bibinfo{person}{Javad~Akbari Torkestani}.} \bibinfo{year}{2011}\natexlab{}.
\newblock \showarticletitle{Connected dominating set problem and its
  application to wireless sensor networks}. In \bibinfo{booktitle}{\emph{The
  First International Conference on Advanced Communications and Computation,
  INFOCOMP}}. \bibinfo{pages}{46--51}.
\newblock


\bibitem[Back et~al\mbox{.}(2002)]%
        {back2002hashcash}
\bibfield{author}{\bibinfo{person}{Adam Back} {et~al\mbox{.}}}
  \bibinfo{year}{2002}\natexlab{}.
\newblock \showarticletitle{Hashcash-a denial of service counter-measure}.
\newblock  (\bibinfo{year}{2002}).
\newblock


\bibitem[Ball et~al\mbox{.}(2017)]%
        {ball2017proofs}
\bibfield{author}{\bibinfo{person}{Marshall Ball}, \bibinfo{person}{Alon
  Rosen}, \bibinfo{person}{Manuel Sabin}, {and}
  \bibinfo{person}{Prashant~Nalini Vasudevan}.}
  \bibinfo{year}{2017}\natexlab{}.
\newblock \showarticletitle{Proofs of useful work}.
\newblock \bibinfo{journal}{\emph{Cryptology ePrint Archive}}
  (\bibinfo{year}{2017}).
\newblock


\bibitem[Boneh et~al\mbox{.}(2006)]%
        {boneh2006strongly}
\bibfield{author}{\bibinfo{person}{Dan Boneh}, \bibinfo{person}{Emily Shen},
  {and} \bibinfo{person}{Brent Waters}.} \bibinfo{year}{2006}\natexlab{}.
\newblock \showarticletitle{Strongly unforgeable signatures based on
  computational Diffie-Hellman}. In \bibinfo{booktitle}{\emph{Public Key
  Cryptography-PKC 2006: 9th International Conference on Theory and Practice in
  Public-Key Cryptography, New York, NY, USA, April 24-26, 2006. Proceedings
  9}}. Springer, \bibinfo{pages}{229--240}.
\newblock


\bibitem[Chleb{\'i}k and Chleb{\'i}kov{\'a}(2004)]%
        {10.1007/978-3-540-30140-0_19}
\bibfield{author}{\bibinfo{person}{Miroslav Chleb{\'i}k} {and}
  \bibinfo{person}{Janka Chleb{\'i}kov{\'a}}.} \bibinfo{year}{2004}\natexlab{}.
\newblock \showarticletitle{Approximation Hardness of Dominating Set Problems}.
  In \bibinfo{booktitle}{\emph{Algorithms -- ESA 2004}},
  \bibfield{editor}{\bibinfo{person}{Susanne Albers} {and}
  \bibinfo{person}{Tomasz Radzik}} (Eds.). \bibinfo{publisher}{Springer Berlin
  Heidelberg}, \bibinfo{address}{Berlin, Heidelberg},
  \bibinfo{pages}{192--203}.
\newblock
\showISBNx{978-3-540-30140-0}


\bibitem[Dotan and Tochner(2020)]%
        {dotan2020proofs}
\bibfield{author}{\bibinfo{person}{Maya Dotan} {and} \bibinfo{person}{Saar
  Tochner}.} \bibinfo{year}{2020}\natexlab{}.
\newblock \showarticletitle{Proofs of Useless Work--Positive and Negative
  Results for Wasteless Mining Systems}.
\newblock \bibinfo{journal}{\emph{arXiv preprint arXiv:2007.01046}}
  (\bibinfo{year}{2020}).
\newblock


\bibitem[Douceur(2002)]%
        {douceur2002sybil}
\bibfield{author}{\bibinfo{person}{John~R Douceur}.}
  \bibinfo{year}{2002}\natexlab{}.
\newblock \showarticletitle{The sybil attack}. In
  \bibinfo{booktitle}{\emph{Peer-to-Peer Systems: First InternationalWorkshop,
  IPTPS 2002 Cambridge, MA, USA, March 7--8, 2002 Revised Papers 1}}. Springer,
  \bibinfo{pages}{251--260}.
\newblock


\bibitem[Fitzi et~al\mbox{.}(2022)]%
        {fitzi2022ofelimos}
\bibfield{author}{\bibinfo{person}{Matthias Fitzi}, \bibinfo{person}{Aggelos
  Kiayias}, \bibinfo{person}{Giorgos Panagiotakos}, {and}
  \bibinfo{person}{Alexander Russell}.} \bibinfo{year}{2022}\natexlab{}.
\newblock \showarticletitle{Ofelimos: Combinatorial Optimization via
  Proof-of-Useful-Work: A Provably Secure Blockchain Protocol}. In
  \bibinfo{booktitle}{\emph{Advances in Cryptology--CRYPTO 2022: 42nd Annual
  International Cryptology Conference, CRYPTO 2022, Santa Barbara, CA, USA,
  August 15--18, 2022, Proceedings, Part II}}. Springer,
  \bibinfo{pages}{339--369}.
\newblock


\bibitem[Fomin et~al\mbox{.}(2009)]%
        {fomin2009measure}
\bibfield{author}{\bibinfo{person}{Fedor~V Fomin}, \bibinfo{person}{Fabrizio
  Grandoni}, {and} \bibinfo{person}{Dieter Kratsch}.}
  \bibinfo{year}{2009}\natexlab{}.
\newblock \showarticletitle{A measure \& conquer approach for the analysis of
  exact algorithms}.
\newblock \bibinfo{journal}{\emph{Journal of the ACM (JACM)}}
  \bibinfo{volume}{56}, \bibinfo{number}{5} (\bibinfo{year}{2009}),
  \bibinfo{pages}{1--32}.
\newblock


\bibitem[Garey and Johnson(1979)]%
        {garey1979computers}
\bibfield{author}{\bibinfo{person}{Michael~R Garey} {and}
  \bibinfo{person}{David~S Johnson}.} \bibinfo{year}{1979}\natexlab{}.
\newblock \showarticletitle{Computers and intractability}.
\newblock \bibinfo{journal}{\emph{A Guide to the}} (\bibinfo{year}{1979}).
\newblock


\bibitem[Gennaro et~al\mbox{.}(2007)]%
        {gennaro2007secure}
\bibfield{author}{\bibinfo{person}{Rosario Gennaro}, \bibinfo{person}{Stanislaw
  Jarecki}, \bibinfo{person}{Hugo Krawczyk}, {and} \bibinfo{person}{Tal
  Rabin}.} \bibinfo{year}{2007}\natexlab{}.
\newblock \showarticletitle{Secure distributed key generation for discrete-log
  based cryptosystems}.
\newblock \bibinfo{journal}{\emph{Journal of Cryptology}}  \bibinfo{volume}{20}
  (\bibinfo{year}{2007}), \bibinfo{pages}{51--83}.
\newblock


\bibitem[Huestis(2023)]%
        {energy}
\bibfield{author}{\bibinfo{person}{Samuel Huestis}.}
  \bibinfo{year}{2023}\natexlab{}.
\newblock \bibinfo{title}{Cryptocurrency’s Energy Consumption Problem}.
\newblock
  \bibinfo{howpublished}{\url{https://rmi.org/cryptocurrencys-energy-consumption-problem/}}.
\newblock


\bibitem[Igor~Barinov(2018)]%
        {poa}
\bibfield{author}{\bibinfo{person}{Pavel~Khahulin Igor~Barinov,
  Viktor~Baranov}.} \bibinfo{year}{2018}\natexlab{}.
\newblock \bibinfo{title}{POA Network Whitepaper}.
\newblock
  \bibinfo{howpublished}{\url{https://github.com/poanetwork/wiki/wiki/POA-Network-Whitepaper}}.
\newblock


\bibitem[Karantias et~al\mbox{.}(2020)]%
        {karantias2020proof}
\bibfield{author}{\bibinfo{person}{Kostis Karantias}, \bibinfo{person}{Aggelos
  Kiayias}, {and} \bibinfo{person}{Dionysis Zindros}.}
  \bibinfo{year}{2020}\natexlab{}.
\newblock \showarticletitle{Proof-of-burn}. In
  \bibinfo{booktitle}{\emph{Financial Cryptography and Data Security: 24th
  International Conference, FC 2020, Kota Kinabalu, Malaysia, February 10--14,
  2020 Revised Selected Papers 24}}. Springer, \bibinfo{pages}{523--540}.
\newblock


\bibitem[King(2013)]%
        {king2013primecoin}
\bibfield{author}{\bibinfo{person}{Sunny King}.}
  \bibinfo{year}{2013}\natexlab{}.
\newblock \showarticletitle{Primecoin: Cryptocurrency with prime number
  proof-of-work}.
\newblock \bibinfo{journal}{\emph{July 7th}} \bibinfo{volume}{1},
  \bibinfo{number}{6} (\bibinfo{year}{2013}).
\newblock


\bibitem[King(2014)]%
        {gap}
\bibfield{author}{\bibinfo{person}{Sunny King}.}
  \bibinfo{year}{2014}\natexlab{}.
\newblock \bibinfo{title}{What is Gapcoin?}
\newblock \bibinfo{howpublished}{\url{https://gapcoin.org}}.
\newblock


\bibitem[King and Nadal(2012)]%
        {king2012ppcoin}
\bibfield{author}{\bibinfo{person}{Sunny King} {and} \bibinfo{person}{Scott
  Nadal}.} \bibinfo{year}{2012}\natexlab{}.
\newblock \showarticletitle{Ppcoin: Peer-to-peer crypto-currency with
  proof-of-stake}.
\newblock \bibinfo{journal}{\emph{self-published paper, August}}
  \bibinfo{volume}{19}, \bibinfo{number}{1} (\bibinfo{year}{2012}).
\newblock


\bibitem[Loe and Quaglia(2018)]%
        {loe2018conquering}
\bibfield{author}{\bibinfo{person}{Angelique~Faye Loe} {and}
  \bibinfo{person}{Elizabeth~A Quaglia}.} \bibinfo{year}{2018}\natexlab{}.
\newblock \showarticletitle{Conquering generals: an np-hard proof of useful
  work}. In \bibinfo{booktitle}{\emph{Proceedings of the 1st Workshop on
  Cryptocurrencies and Blockchains for Distributed Systems}}.
  \bibinfo{pages}{54--59}.
\newblock


\bibitem[Maleš et~al\mbox{.}(2023)]%
        {sym15010140}
\bibfield{author}{\bibinfo{person}{Uroš Maleš}, \bibinfo{person}{Dušan
  Ramljak}, \bibinfo{person}{Tatjana Jakšić~Krüger},
  \bibinfo{person}{Tatjana Davidović}, \bibinfo{person}{Dragutin Ostojić},
  {and} \bibinfo{person}{Abhay Haridas}.} \bibinfo{year}{2023}\natexlab{}.
\newblock \showarticletitle{Controlling the Difficulty of Combinatorial
  Optimization Problems for Fair Proof-of-Useful-Work-Based Blockchain
  Consensus Protocol}.
\newblock \bibinfo{journal}{\emph{Symmetry}} \bibinfo{volume}{15},
  \bibinfo{number}{1} (\bibinfo{year}{2023}).
\newblock
\showISSN{2073-8994}
\urldef\tempurl%
\url{https://doi.org/10.3390/sym15010140}
\showDOI{\tempurl}


\bibitem[Moran and Orlov(2019)]%
        {moran2019simple}
\bibfield{author}{\bibinfo{person}{Tal Moran} {and} \bibinfo{person}{Ilan
  Orlov}.} \bibinfo{year}{2019}\natexlab{}.
\newblock \showarticletitle{Simple proofs of space-time and rational proofs of
  storage}. In \bibinfo{booktitle}{\emph{Advances in Cryptology--CRYPTO 2019:
  39th Annual International Cryptology Conference, Santa Barbara, CA, USA,
  August 18--22, 2019, Proceedings, Part I 39}}. Springer,
  \bibinfo{pages}{381--409}.
\newblock


\bibitem[Nakamoto(2008)]%
        {nakamoto2008bitcoin}
\bibfield{author}{\bibinfo{person}{Satoshi Nakamoto}.}
  \bibinfo{year}{2008}\natexlab{}.
\newblock \showarticletitle{Bitcoin: A peer-to-peer electronic cash system}.
\newblock \bibinfo{journal}{\emph{Decentralized business review}}
  (\bibinfo{year}{2008}), \bibinfo{pages}{21260}.
\newblock


\bibitem[Nicolas et~al\mbox{.}(2019)]%
        {nicolas2019comprehensive}
\bibfield{author}{\bibinfo{person}{Kervins Nicolas}, \bibinfo{person}{Yi Wang},
  {and} \bibinfo{person}{George~C Giakos}.} \bibinfo{year}{2019}\natexlab{}.
\newblock \showarticletitle{Comprehensive overview of selfish mining and double
  spending attack countermeasures}. In \bibinfo{booktitle}{\emph{2019 IEEE 40th
  Sarnoff Symposium}}. IEEE, \bibinfo{pages}{1--6}.
\newblock


\bibitem[Oliver et~al\mbox{.}(2017)]%
        {oliver2017proposal}
\bibfield{author}{\bibinfo{person}{Carlos~G Oliver},
  \bibinfo{person}{Alessandro Ricottone}, {and} \bibinfo{person}{Pericles
  Philippopoulos}.} \bibinfo{year}{2017}\natexlab{}.
\newblock \showarticletitle{Proposal for a fully decentralized blockchain and
  proof-of-work algorithm for solving NP-complete problems}.
\newblock \bibinfo{journal}{\emph{arXiv preprint arXiv:1708.09419}}
  (\bibinfo{year}{2017}).
\newblock


\bibitem[Park et~al\mbox{.}(2018)]%
        {park2018spacemint}
\bibfield{author}{\bibinfo{person}{Sunoo Park}, \bibinfo{person}{Albert Kwon},
  \bibinfo{person}{Georg Fuchsbauer}, \bibinfo{person}{Peter Ga{\v{z}}i},
  \bibinfo{person}{Jo{\"e}l Alwen}, {and} \bibinfo{person}{Krzysztof
  Pietrzak}.} \bibinfo{year}{2018}\natexlab{}.
\newblock \showarticletitle{Spacemint: A cryptocurrency based on proofs of
  space}. In \bibinfo{booktitle}{\emph{Financial Cryptography and Data
  Security: 22nd International Conference, FC 2018, Nieuwpoort, Cura{\c{c}}ao,
  February 26--March 2, 2018, Revised Selected Papers 22}}. Springer,
  \bibinfo{pages}{480--499}.
\newblock


\bibitem[Philippopoulos et~al\mbox{.}(2020)]%
        {philippopoulos2020difficulty}
\bibfield{author}{\bibinfo{person}{Pericles Philippopoulos},
  \bibinfo{person}{Alessandro Ricottone}, {and} \bibinfo{person}{Carlos~G
  Oliver}.} \bibinfo{year}{2020}\natexlab{}.
\newblock \showarticletitle{Difficulty Scaling in Proof of Work for
  Decentralized Problem Solving}.
\newblock \bibinfo{journal}{\emph{Ledger}}  \bibinfo{volume}{5}
  (\bibinfo{year}{2020}).
\newblock


\bibitem[Platt et~al\mbox{.}(2021)]%
        {9741872}
\bibfield{author}{\bibinfo{person}{Moritz Platt}, \bibinfo{person}{Johannes
  Sedlmeir}, \bibinfo{person}{Daniel Platt}, \bibinfo{person}{Jiahua Xu},
  \bibinfo{person}{Paolo Tasca}, \bibinfo{person}{Nikhil Vadgama}, {and}
  \bibinfo{person}{Juan~Ignacio Ibañez}.} \bibinfo{year}{2021}\natexlab{}.
\newblock \showarticletitle{The Energy Footprint of Blockchain Consensus
  Mechanisms Beyond Proof-of-Work}. In \bibinfo{booktitle}{\emph{2021 IEEE 21st
  International Conference on Software Quality, Reliability and Security
  Companion (QRS-C)}}. \bibinfo{pages}{1135--1144}.
\newblock
\urldef\tempurl%
\url{https://doi.org/10.1109/QRS-C55045.2021.00168}
\showDOI{\tempurl}


\bibitem[Rao et~al\mbox{.}(2020)]%
        {sym12111885}
\bibfield{author}{\bibinfo{person}{Yongsheng Rao}, \bibinfo{person}{Saeed
  Kosari}, \bibinfo{person}{Zehui Shao}, \bibinfo{person}{Ruiqi Cai}, {and}
  \bibinfo{person}{Liu Xinyue}.} \bibinfo{year}{2020}\natexlab{}.
\newblock \showarticletitle{A Study on Domination in Vague Incidence Graph and
  Its Application in Medical Sciences}.
\newblock \bibinfo{journal}{\emph{Symmetry}} \bibinfo{volume}{12},
  \bibinfo{number}{11} (\bibinfo{year}{2020}).
\newblock
\showISSN{2073-8994}
\urldef\tempurl%
\url{https://doi.org/10.3390/sym12111885}
\showDOI{\tempurl}


\bibitem[Ren(2019)]%
        {ren2019analysis}
\bibfield{author}{\bibinfo{person}{Ling Ren}.} \bibinfo{year}{2019}\natexlab{}.
\newblock \showarticletitle{Analysis of nakamoto consensus}.
\newblock \bibinfo{journal}{\emph{Cryptology ePrint Archive}}
  (\bibinfo{year}{2019}).
\newblock


\bibitem[Seshadhri et~al\mbox{.}(2012)]%
        {seshadhri2012community}
\bibfield{author}{\bibinfo{person}{Comandur Seshadhri},
  \bibinfo{person}{Tamara~G Kolda}, {and} \bibinfo{person}{Ali Pinar}.}
  \bibinfo{year}{2012}\natexlab{}.
\newblock \showarticletitle{Community structure and scale-free collections of
  Erd{\H{o}}s-R{\'e}nyi graphs}.
\newblock \bibinfo{journal}{\emph{Physical Review E}} \bibinfo{volume}{85},
  \bibinfo{number}{5} (\bibinfo{year}{2012}), \bibinfo{pages}{056109}.
\newblock


\bibitem[Syafruddin et~al\mbox{.}(2019)]%
        {syafruddin2019blockchain}
\bibfield{author}{\bibinfo{person}{Willa~Ariela Syafruddin},
  \bibinfo{person}{Sajjad Dadkhah}, {and} \bibinfo{person}{Mario K{\"o}ppen}.}
  \bibinfo{year}{2019}\natexlab{}.
\newblock \showarticletitle{Blockchain scheme based on evolutionary proof of
  work}. In \bibinfo{booktitle}{\emph{2019 IEEE Congress on Evolutionary
  Computation (CEC)}}. IEEE, \bibinfo{pages}{771--776}.
\newblock


\bibitem[Tasca and Tessone(2017)]%
        {tasca2017taxonomy}
\bibfield{author}{\bibinfo{person}{Paolo Tasca} {and}
  \bibinfo{person}{Claudio~J Tessone}.} \bibinfo{year}{2017}\natexlab{}.
\newblock \showarticletitle{Taxonomy of blockchain technologies. Principles of
  identification and classification}.
\newblock \bibinfo{journal}{\emph{arXiv preprint arXiv:1708.04872}}
  (\bibinfo{year}{2017}).
\newblock


\bibitem[Todorovi{\'c} et~al\mbox{.}(2022)]%
        {todorovic2022proof}
\bibfield{author}{\bibinfo{person}{Milan Todorovi{\'c}}, \bibinfo{person}{Luka
  Matijevi{\'c}}, \bibinfo{person}{Du{\v{s}}an Ramljak},
  \bibinfo{person}{Tatjana Davidovi{\'c}}, \bibinfo{person}{Dragan
  Uro{\v{s}}evi{\'c}}, \bibinfo{person}{Tatjana Jak{\v{s}}i{\'c}~Kr{\"u}ger},
  {and} \bibinfo{person}{{\DJ}or{\dj}e Jovanovi{\'c}}.}
  \bibinfo{year}{2022}\natexlab{}.
\newblock \showarticletitle{Proof-of-Useful-Work: BlockChain Mining by Solving
  Real-Life Optimization Problems}.
\newblock \bibinfo{journal}{\emph{Symmetry}} \bibinfo{volume}{14},
  \bibinfo{number}{9} (\bibinfo{year}{2022}), \bibinfo{pages}{1831}.
\newblock


\bibitem[Wang(2014)]%
        {wang2014domination}
\bibfield{author}{\bibinfo{person}{Guangyuan Wang}.}
  \bibinfo{year}{2014}\natexlab{}.
\newblock \emph{\bibinfo{title}{Domination problems in social networks}}.
\newblock \bibinfo{thesistype}{Ph.\,D. Dissertation}.
  \bibinfo{school}{University of Southern Queensland}.
\newblock


\bibitem[Xiao et~al\mbox{.}(2020)]%
        {xiao2020survey}
\bibfield{author}{\bibinfo{person}{Yang Xiao}, \bibinfo{person}{Ning Zhang},
  \bibinfo{person}{Wenjing Lou}, {and} \bibinfo{person}{Y~Thomas Hou}.}
  \bibinfo{year}{2020}\natexlab{}.
\newblock \showarticletitle{A survey of distributed consensus protocols for
  blockchain networks}.
\newblock \bibinfo{journal}{\emph{IEEE Communications Surveys \& Tutorials}}
  \bibinfo{volume}{22}, \bibinfo{number}{2} (\bibinfo{year}{2020}),
  \bibinfo{pages}{1432--1465}.
\newblock


\end{thebibliography}

\end{document}